\documentclass[letterpaper,11pt,onecolumn,draftcls]{IEEEtran}  % Comment this line out

\usepackage{epsfig} % for postscript graphics files
\usepackage{amsmath} % assumes amsmath package installed
\usepackage{amssymb}  % assumes amsmath package installed

\usepackage{cite}      % Written by Donald Arseneau
\usepackage{subfigure} % Written by Steven Douglas Cochran
\usepackage{url}       % Written by Donald Arseneau
\usepackage{verbatim}   % multi-line comment: \begin{comment},\end{comment}
\usepackage{tabularx}

\usepackage{enumerate}
\usepackage{algorithm}
\usepackage{algpseudocode}

\usepackage{bbold}
\usepackage[dvips]{color}
\newcommand{\high}[1]{{\color{black}{#1}}}

\newcommand{\HRule}{\rule{\linewidth}{0.3mm}}

\begin{document}

\title{%\LARGE \bf
Delay-Based Back-Pressure Scheduling in Multihop Wireless Networks}

%\author{ \parbox{3 in}{\centering Huibert Kwakernaak*
%         \thanks{*Use the $\backslash$thanks command to put information here}\\
%         Faculty of Electrical Engineering, Mathematics and Computer Science\\
%         University of Twente\\
%         7500 AE Enschede, The Netherlands\\
%         {\tt\small h.kwakernaak@autsubmit.com}}
%         \hspace*{ 0.5 in}
%         \parbox{3 in}{ \centering Pradeep Misra**
%         \thanks{**The footnote marks may be inserted manually}\\
%        Department of Electrical Engineering \\
%         Wright State University\\
%         Dayton, OH 45435, USA\\
%         {\tt\small pmisra@cs.wright.edu}}
%}
%\author{Bo Ji, Changhee Joo and Ness B. Shroff}

% \author{Bo~Ji, Changhee~Joo, and Ness~B.~Shroff%
% \thanks{B. Ji is with Department of ECE at the 
% Ohio State University. 
% C. Joo is with School of ECE at UNIST, Korea.
% N. B. Shroff is with Departments of ECE 
% and CSE at the Ohio State University.}%
% \thanks{Emails: ji@ece.osu.edu, cjoo@unist.ac.kr, shroff@ece.osu.edu.}
% \thanks{A preliminary version of this work was presented at IEEE INFOCOM 2011.}
% }%
\author{Bo~Ji,~\IEEEmembership{Student~Member,~IEEE,} 
Changhee~Joo,~\IEEEmembership{Member,~IEEE,}
and Ness~B.~Shroff,~\IEEEmembership{Fellow,~IEEE}%
\thanks{B. Ji is with Department of ECE at the 
Ohio State University. 
C. Joo is with School of ECE at UNIST, Korea.
N. B. Shroff is with Departments of ECE 
and CSE at the Ohio State University.}%
\thanks{Emails: ji@ece.osu.edu, cjoo@unist.ac.kr, shroff@ece.osu.edu.}
\thanks{This work was supported in part by ARO MURI Award W911NF-08-1-0238, 
and NSF Awards 1012700-CNS, 0721236-CNS, 0721434-CNS, and 1065136-CNS, and in part by 
the Basic Science Research Program through the National Research Foundation 
of Korea (NRF), funded by the Ministry of Education, Science, and Technology 
(No. 2012-0003227).}
\thanks{A preliminary version of this work was presented at IEEE INFOCOM 2011.}
}%

\newcounter{theorem}
\newtheorem{theorem}{Theorem}
\newtheorem{proposition}[theorem]{\it Proposition}
\newtheorem{lemma}[theorem]{\it Lemma}
\newtheorem{corollary}[theorem]{\it Corollary}
\newcounter{definition}
\newtheorem{definition}{\it Definition}

\newcommand{\Graph}{\mathcal{G}}
\newcommand{\Node}{\mathcal{V}}
\newcommand{\Vertex}{\mathcal{V}}
\newcommand{\Edge}{\mathcal{E}}
\newcommand{\Link}{\mathcal{L}}
\newcommand{\Flow}{\mathcal{S}}
\newcommand{\Route}{\mathcal{H}}
\newcommand{\Hop}{{H}}
\newcommand{\Matching}{\mathcal{M}}
\newcommand{\Pair}{\mathcal{P}}
\newcommand{\Int}{\mathbb{Z}}
\newcommand{\Expect}{\mathbf{E}}
\newcommand{\System}{\mathcal{Y}}
\newcommand{\Markov}{\mathcal{X}}
\newcommand{\vM}{\vec{M}}
\newcommand{\bM}{\mathbf{M}}
\newcommand{\s}{{(s)}}
\newcommand{\sk}{{s,k}}
\newcommand{\hsk}{{\hat{s},\hat{k}}}
\newcommand{\rj}{{r,j}}
\newcommand{\xn}{{x_n}}
\newcommand{\xnj}{{x_{n_j}}}
\newcommand{\UGraph}{{U}}
\newcommand{\UVertex}{{X}}
\newcommand{\UEdge}{{Y}}
\newcommand{\vlambda}{\vec{\lambda}}
\newcommand{\vpi}{\vec{\pi}}
\newcommand{\vphi}{\vec{\phi}}
\newcommand{\vpsi}{\vec{\psi}}
\newcommand{\blambda}{\mathbf{\lambda}}
\newcommand{\tM}{\tilde{M}}
\newcommand{\tpi}{\tilde{\pi}}
\newcommand{\tphi}{\tilde{\phi}}
\newcommand{\tpsi}{\tilde{\psi}}
\newcommand{\vmu}{\vec{\mu}}
\newcommand{\vnu}{\vec{\nu}}
\newcommand{\valpha}{\vec{\alpha}}
\newcommand{\vbeta}{\vec{\beta}}
\newcommand{\ve}{\vec{e}}
\newcommand{\vxi}{\vec{\xi}}
\newcommand{\hgamma}{\gamma}
\newcommand{\GMSLambda}{\Lambda_{\text{\it GMS}}}
\newcommand{\argmax}{\operatornamewithlimits{argmax}}
\newcommand{\card}{\aleph}
\newcommand{\V}{\mathcal{V}}
\newcommand{\X}{\mathcal{X}}
\newcommand{\Y}{\mathcal{Y}}

\maketitle
%\thispagestyle{empty}
%\pagestyle{empty}

%%%%%%%%%%%%%%%%%%%%%%%%%%%%%%%%%%%%%%%%
\begin{abstract}
%%%%%%%%%%%%%%%%%%%%%%%%%%%%%%%%%%%%%%%%
Scheduling is a critical and challenging resource allocation 
mechanism for multihop wireless networks. It is well known 
that scheduling schemes that favor links with larger queue 
length can achieve high throughput performance. However, 
these queue-length-based schemes could potentially suffer 
from large (even infinite) packet delays due to the well-known 
\emph{last packet problem}, whereby packets belonging to some
flows may be excessively delayed due to lack of subsequent 
packet arrivals. Delay-based schemes have the potential to 
resolve this last packet problem by scheduling the link based 
on the delay the packet has encountered. However, characterizing
throughput-optimality of these delay-based schemes has largely 
been an open problem in multihop wireless networks (except in 
limited cases where the traffic is single-hop.) In this paper, 
we investigate delay-based scheduling schemes for multihop 
traffic scenarios \high{with fixed routes. We develop a scheduling 
scheme based on a new delay metric, and} show that the proposed 
scheme achieves optimal throughput performance. Further, we 
conduct simulations to support our analytical results, and 
show that the delay-based scheduler successfully removes 
excessive packet delays, while it achieves the same throughput 
region as the queue-length-based scheme.

\end{abstract}

%%%%%%%%%%%%%%%%%%%%%%%%%%%%%%%%%%%%%%%%%%%%
\section{Introduction} \label{sec:intro}
%%%%%%%%%%%%%%%%%%%%%%%%%%%%%%%%%%%%%%%%%%%%
Link scheduling is a critical resource allocation component in multihop 
wireless networks, and also perhaps the most challenging. \high{The seminal work 
of \cite{tassiulas92} introduces a joint adaptive routing and scheduling 
algorithm, called Queue-length-based Back-Pressure (Q-BP), that has been 
shown to be throughput-optimal, i.e., it can stabilize the network under 
any feasible load. This paper focuses on the settings with fixed routes, 
where the Q-BP algorithm becomes a scheduling algorithm.}
Since the development of Q-BP, there have been numerous extensions
that have integrated it in an overall optimal cross-layer framework.
Further, easier-to-implement queue-length-based scheduling schemes
have been developed and shown to be throughput-efficient (see \cite{lin06c}
and references therein). Some recent attempts \cite{warrier09, sridharan09,
moeller10} focus on designing real-world wireless protocols using the 
ideas behind these algorithms.

While these queue-length-based schedulers have been shown to achieve
excellent throughput performance, they are usually evaluated under
the assumption that flows have an infinite amount of data and keep
injecting packets into the network. However, in practice, when accounting
for multiple time scales \cite{van09, liu10a, liu10b}, there also 
exist other types of flows that have a finite number of packets to 
transmit, which can result in the well-known \emph{last packet problem}: 
consider a queue that holds the last packet of a flow, then the packet 
does not see any subsequent packet arrivals, and thus the queue length 
remains very small and the link may be starved for a long time, since 
the queue-length-based schemes give a higher priority to links with a 
larger queue length. In such a scenario with flow-level dynamics, it 
has also been shown in \cite{van09} that the queue-length-based schemes 
may not even be throughput-optimal.

Recent works in \cite{mekkittikul96, andrews01, shakkottai02, andrews04, 
sadiq09, neely10} have studied the performance of delay-based scheduling 
algorithms that use Head-of-Line (HOL) delays instead of queue lengths 
as link weights.
One desirable property of the delay-based approach is that they provide 
an intuitive way around the last packet problem. The schedulers give a 
higher priority to the links with a larger weight as before, but now the 
weight (i.e., the HOL delay) of a link increases with time until the link 
is scheduled. Hence, if the link with the last packet is not scheduled 
at this moment, it is more likely to be scheduled in the next time.
However, the throughput of the delay-based scheduling schemes is not 
fully understood, and has only been established for limited cases 
with single-hop traffic.

The delay-based approach was introduced in \cite{mekkittikul96} for 
scheduling in Input-Queued switches. The results have been extended to
wireless networks for single-hop traffic, providing throughput-optimal 
delay-based MaxWeight scheduling algorithms \cite{shakkottai02,andrews04,
eryilmaz05}.
It has also been shown that delay-based schemes with appropriately chosen 
weight parameters provide good Quality of Service (QoS) \cite{andrews01},
and can be used as an important component in a cross-layer protocol design 
\cite{neely10}.
The performance of the delay-based MaxWeight scheduler has been further 
investigated in a single-hop network with flow-level dynamics \cite{sadiq09}. 
The results show that, when flows arrive at the base station carrying a 
finite amount of data, the delay-based MaxWeight scheduler achieves
optimal throughput performance while its queue-length-based counterpart 
does not.

\high{It should be noted that even for the multihop wireless networks with 
fixed routes, the scheduling problem is both important and challenging. There 
are many existing works focusing on such scenarios with fixed routes (see 
\cite{liuj09,bui11,gupta11} for examples).} \emph{However, in multihop 
wireless networks, the throughput performance of these delay-based schemes 
has largely been an open problem.} To the best of our knowledge, \high{even 
with the assumption of fixed routes,} there are no prior works that employ 
delay-based algorithms to address the important issue of throughput-optimal 
scheduling in multihop wireless networks. Indeed, the problem becomes much 
more challenging in the multihop scenario. In \cite{andrews04}, the key 
idea in showing throughput-optimality of the delay-based MaxWeight scheduler 
is to exploit the following property: after a finite time, there exists 
a linear relation between queue lengths and HOL delays in the fluid 
limits (which we formally define in Section~\ref{subsec:fluid}), where 
the ratio is the mean arrival rate. Hence, the delay-based MaxWeight 
scheme is basically equivalent to its queue-length-based counterpart, and 
thus achieves the optimal throughput. This property holds for the single-hop 
traffic. Since given that the exogenous arrival processes follow the Strong 
Law of Large Numbers (SLLN) and the fluid limits exist, the arrival processes 
are deterministic with constant rates in the fluid limits. \emph{However, 
such a linear relation does not necessarily hold for the multihop traffic, 
since at a non-source (or relay) node, the arrival process may not satisfy 
SLLN and the packet arrival rate may not even be a constant, depending on 
the underlying scheduler's dynamics. }
To this end, we investigate delay-based scheduling schemes that achieve 
optimal throughput performance in multihop wireless networks.

Unlike previous delay-based schemes, we view the packet delay as a sojourn 
time in the network, and re-design the delay metric of the queue as the 
sojourn-time difference between the queue's HOL packet and the HOL packet 
of its previous hop (see Eq.~(\ref{eq:linkd}) for the formal definition). 
Using this new metric, we can establish a linear relation between queue 
lengths and delays in the fluid limits. The linear relation then plays 
the key role in showing that the proposed Delay-based Back-Pressure (D-BP) 
scheduling scheme is throughput-optimal in multihop networks.

In summary, the main contributions of our paper are as follows:
\begin{itemize}

\high{
% \item We first re-visit throughput-optimality of Q-BP using fluid 
% limit techniques, which will later be extended to analyze D-BP. 
% \item We devise a new delay metric for multihop wireless networks 
% and develop the D-BP algorithm, under which a linear relation between 
% queue lengths and delays in the fluid limits can be established. 
% From this linear relation, we can show that D-BP achieves optimal 
% throughput performance. Further, we develop a simpler greedy 
% approximation of D-BP for practical implementation.

\item We devise a new delay metric for multihop wireless networks 
and develop the D-BP algorithm, under which a linear relation between 
queue lengths and delays in the fluid limits can be established. From 
this linear relation, we can show that D-BP achieves optimal throughput 
performance. To do this, we first re-visit throughput-optimality of 
Q-BP using fluid limit techniques. Further, we develop a simpler greedy 
approximation of D-BP for practical implementation.

%
% \item We conduct simulations to evaluate the performance of 
% delay-based schedulers. Through simulations, we observe that 
% the last packet problem can cause excessive delays for certain 
% flows under Q-BP, while the problem is eliminated under D-BP. 
% Further, in the case of Q-BP, even though the average delays 
% experienced in the network may be similar to D-BP, the tail 
% of the delay distribution could be substantially longer. We 
% also show that, D-BP can not only achieve the same throughput 
% region as Q-BP, but also guarantee better fairness by scheduling 
% the links based on delays and not starving certain flows that 
% lack subsequent packet arrivals (or have very large inter-arrival 
% times between groups of packet arrivals). Further, being interested 
% in simpler algorithms that are amenable to practical implementations, 
% we simulate greedy approximation algorithms to Q-BP and D-BP. 
% In the results, we observe that in certain scenarios, delay-based 
% greedy approximation achieves a larger throughput region than its 
% queue-length-based counterpart.
%
\item We provide extensive simulation results to evaluate the 
performance of the delay-based schedulers, including D-BP. 
Through simulations, i) we observe that the last packet problem 
can cause excessive delays for certain flows under Q-BP, while 
the problem is eliminated under D-BP. ii) We show that D-BP 
also achieves better fairness and prevents the flows that lack 
subsequent packet arrivals from starving. iii) Finally, we 
simulate the simpler greedy approximation algorithms of Q-BP and 
D-BP, and show that the delay-based approximation empirically 
achieves a throughput region that is no smaller than that of 
its queue-based counterpart. 
}
\end{itemize}

The paper is organized as follows. In Section~\ref{sec:model}, 
we present a detailed description of our system model. In 
Section~\ref{sec:qbp}, we show throughput-optimality of Q-BP 
using fluid limit techniques, and extend the analysis to D-BP 
in Section~\ref{sec:dbp}. The discussions are further extended 
to the greedy algorithms in Section~\ref{sec:greedy}. We evaluate 
the performance of delay-based schedulers through simulations 
in Section~\ref{sec:simulation}, and conclude our paper in 
Section~\ref{sec:conclusion}. 

%%%%%%%%%%%%%%%%%%%%%%%%%%%%%%%%%%%%%%%%%%%
\section{System Model} \label{sec:model}
%%%%%%%%%%%%%%%%%%%%%%%%%%%%%%%%%%%%%%%%%%%
We consider a multihop wireless network described by a directed graph 
$\Graph=(\Vertex,\Edge)$, where $\Vertex$ denotes the set of nodes and 
$\Edge$ denotes the set of links. Nodes are wireless transmitters/receivers 
and links are wireless channels between two nodes if they can directly 
communicate with each other. During a single time slot, multiple links 
that do not interfere with each other can be active at the same time, 
and each active link transmits one packet during the time slot if its 
queue is not empty. Let $\Flow$ denote the set of flows in the network. 
We assume that each flow has a single, fixed, and loop-free route. The 
route of flow $s$ has an $\Hop(s)$-hop length from the source to the 
destination, where each $k$-th hop link is denoted by $(\sk)$. Let
$\Hop^{\max} \triangleq \max_{s \in \Flow} \Hop(s) < \infty$ denote the 
length of the longest route over all flows. Note that the assumption 
of single route and unit link capacity is only for ease of exposition, 
and one can readily extend the results to more general scenarios with 
\emph{multiple fixed routes and heterogeneous link rates}, applying 
the techniques used in this paper. To specify wireless interference, 
we consider the $k$-th hop of each flow $s$ or link-flow-pair $(\sk)$. 
Let $\Pair$ denote the set of all link-flow-pairs, i.e.,
\[
\Pair \triangleq \{(\sk) ~|~ s \in \Flow,~ 1 \le k \le \Hop(s)\}.
\]  
The set of link-flow-pairs that interfere with $(\sk)$ can be described as
\begin{equation}
\label{eq:interference}
\begin{split}
I(\sk) \triangleq \{ (\rj) \in \Pair ~|~ &(\rj) ~\text{interferes with}~ (\sk), \\
& \text{or}~ (\rj) = (\sk) \}. 
\end{split}
\end{equation}
Note that the interference model we adopt is very general, and includes the 
class of the \emph{$K$-hop} interference model\footnote{Under the $K$-hop 
interference model, two links within a $K$-hop ``distance" interfere with 
each other and cannot be activated at the same time \cite{sharma06}. When 
$K=1$, it is also called the \emph{primary} or \emph{node-exclusive} 
interference model. The 1-hop interference model has been known as a good 
representation for Bluetooth or FH-CDMA networks \cite{hajek88, joo09, 
joo09c, lin06}. When $K=2$, it is often used to model the ubiquitous IEEE 
802.11 DCF (Distributed Coordination Function) wireless networks 
\cite{chaporkar08, wu07, leconte09, joo09c}.}. A schedule is a set of 
(active or inactive) link-flow-pairs, and can be represented by a vector 
$\vM \in \{0,1\}^{|\Pair|}$, where $|\cdot|$ denotes the cardinality of 
a set. Each element $M_{\sk}$ is set to 1 if link-flow-pair $(\sk)$ is 
active, and 0 if link-flow-pair $(\sk)$ is inactive. Slightly abusing 
the notation, we also use $M$ to denote the set of active link-flow-pairs 
of $\vM$, i.e., $M\triangleq\{(\sk) \in \Pair ~|~ M_\sk=1 \}$. A schedule 
$\vM$ is said to be \emph{feasible} if no two link-flow-pairs of $\vM$ 
interfere with each other, i.e., $(\rj) \notin I(\sk)$ for all $(\rj)$, 
$(\sk)$ with $M_\rj = 1$ and $M_\sk = 1$. Let $\Matching_{\Pair}$ denote 
the set of all feasible schedules in $\Pair$, and let $Co(\Matching_{\Pair})$ 
denote its convex hull. 

Let $A_s(t)$ denote the number of packet arrivals at the source node of flow 
$s$ at time slot $t$.  \high{We assume that packets are of unit length. Similar 
to \cite{andrews04}, we assume that each arrival process $A_s(t)$ is a stationary 
and ergodic Markov chain with countable state space, and satisfies} the Strong 
Law of Large Numbers (SLLN): That is, with probability one,
\begin{equation}
\label{eq:slln}
\textstyle \lim_{t \rightarrow \infty} \frac {\sum^{t-1}_{\tau=0} A_s(\tau)} {t} = \lambda_s,
\end{equation}
for each flow $s \in \Flow$, where $\lambda_s$ denotes the mean arrival rate 
of flow $s$. We let $\vlambda \triangleq [\lambda_1,\lambda_2, \cdots, 
\lambda_{|\Flow|}]$ denote the arrival rate vector. 
% Assumption (\ref{eq:slln}) on arrival processes is mild. It is satisfied, 
% for example, when the number of arrivals at each time slot is \emph{i.i.d} 
% across time with mean rates $\vlambda$.

Let $Q_\sk(t)$ denote the number of packets at the queue of $(\sk)$ at 
the beginning of time slot $t$. For notational ease, we also use $Q_\sk$ 
to denote the queue itself. We let $\vec{Q}(t)\triangleq [Q_\sk(t),~(\sk) 
\in \Pair]$ denote the queue length vector at time slot $t$, and use 
$\|\cdot\|$ to denote the $L_1$-norm of a vector, e.g., $\|\vec{Q}(t)\|=
\sum_{(\sk)\in\Pair} Q_\sk(t)$. Let $\Pi_\sk(t)$ denote the service of 
$Q_\sk$ at time slot $t$, which takes a value of either 1 if link-flow-pair 
$(\sk)$ is active, or 0 otherwise, in our settings. We let $\Psi_\sk(t)$
denote the actual number of packets transmitted from $Q_\sk$ at time slot $t$. 
Clearly, we have $\Psi_\sk(t) \le \Pi_\sk(t)$ for all time slots $t\ge0$. 
Let $P_\sk(t) \triangleq \sum_{i=1}^{k} Q_{s,i}(t)$ denote the cumulative
queue lengths up to the $k$-th hop for flow $s$. By convention, we set 
$Q_{s,\Hop(s)+1}(t) = 0$, and then we have $P_{s,\Hop(s)+1}(t) = P_{s,\Hop(s)}(t)$. 
The queue length evolves according to the following equations:
\begin{equation} 
\label{eq:q_evolution}
Q_\sk(t+1) = Q_\sk(t) + \Psi_{s,k-1}(t) - \Psi_\sk(t),
\end{equation}
where we set $\Psi_{s,0}(t) = A_s(t)$.

Let $F_s(t)$ be the total number of packets that arrive at the source node 
of flow $s$ until time slot $t \ge 0$, including those present at time slot 
0, and let $\hat{F}_{\sk}(t)$ be the total number of packets that are served 
at $Q_{\sk}$ until time slot $t \ge 0$. By convention, we set $\hat{F}_{\sk}(0) 
= 0$ for all link-flow-pairs $(\sk) \in \Pair$. We let $Z_{\sk,i}(t)$ denote 
the sojourn time of the $i$-th packet of $Q_{\sk}$ in the network at time slot 
$t$, where the time is measured from the time when the packet arrives in the 
network (i.e., when the packet arrives at the source node), and let 
$W_{\sk}(t)=Z_{\sk,1}(t)$ denote the sojourn time of the HOL packet of $Q_{\sk}$ 
in the network at time slot $t$. We set $W_{s,0}(t) = 0$ for all $s \in \Flow$. 
Further, if $Q_\sk(t) = 0$, we set $W_{\sk}(t) = W_{s,k-1}(t)$. Letting 
$U_\sk(t) \triangleq t - W_\sk(t)$ denote the time when the HOL packet of $Q_\sk$ 
arrives in the network, we have that 
\begin{equation}
\label{eq:u}
U_{\sk}(t) = \inf \{ \tau \le t ~|~ F_s(\tau) > \hat{F}_\sk(t) \},
~\text{for all}~ t \ge 0.
\end{equation} 

As in \cite{bramson08}, a discrete-time queueing system is said to be 
\emph{stable}, if the underlying Markov chain is \emph{positive Harris 
recurrent}. When the state space is countable and all states communicate
(as in the system that we consider in this paper), this is equivalent to 
the Markov chain being \emph{positive recurrent}. The \emph{throughput 
region} of a scheduling policy is defined as the set of arrival rate 
vectors for which the network remains stable under this policy. Further, 
the \emph{optimal throughput region} (or \emph{stability region}) is 
defined as the union of the throughput regions of all possible scheduling 
policies. We let $\Lambda^{*}$ denote the optimal throughput region, 
which can be represented as
\begin{equation}
\label{eq:otr}
\Lambda^{*} \triangleq \{ \vlambda ~|~ \exists  \vphi \in Co(\Matching_{\Pair}) 
~\text{s.t.}~ \lambda_s \le \phi_\sk, \forall (\sk) \in \Pair \}.
\end{equation}
An arrival rate vector is strictly inside $\Lambda^{*}$, if the inequalities 
above are all strict.

We summarize the notations in Appendix~\ref{app:notation} for quick reference.

%%%%%%%%%%%%%%%%%%%%%%%%%%%%%%%%%%%%%%%%%%%%%%%%%%%%%%%%%%%%%%%%%%%%%%%%%%
\section{Queue-length-based Back-Pressure Algorithm} \label{sec:qbp}
%%%%%%%%%%%%%%%%%%%%%%%%%%%%%%%%%%%%%%%%%%%%%%%%%%%%%%%%%%%%%%%%%%%%%%%%%%
It has been shown in \cite{tassiulas92} that Q-BP stabilizes the network for any 
feasible arrival rate vector using stochastic Lyapunov techniques. Specifically, 
we can use a quadratic Lyapunov function to show that the function has a negative 
drift under Q-BP when queue lengths are large enough. In this section, we re-visit 
throughput-optimality of Q-BP using fluid limit techniques. The analysis will be 
extended later to prove throughput-optimality of the delay-based back-pressure 
algorithm.

To begin with, we define the \emph{queue differential} $\Delta Q_\sk(t)$ as
\begin{equation} 
\label{eq:qd}
\Delta Q_\sk(t) \triangleq Q_\sk(t) - Q_{s,k+1}(t),
\end{equation}
and specify the back-pressure algorithm based on queue lengths as follows. 

\noindent {\bf Queue-length-based Back-Pressure (Q-BP) algorithm:}
\begin{equation}
\label{eq:qbp}
\textstyle \vec{M}^* \in \argmax_{\vec{M} \in \Matching_\Pair} \sum_{(\sk) \in \Pair} 
\Delta Q_\sk(t) \cdot M_\sk.
\end{equation}
The algorithm needs to solve a MaxWeight problem with weights as queue differentials,
and ties can be broken arbitrarily if there is more than one schedule that has the 
largest weight sum. 

We establish the fluid limits of the system in the following subsection.

%%%%%%%%%%%%%%%%%%%%%%%%%%%%%%%%%%%%%%%%%%%%%%
\subsection{Fluid Limits} \label{subsec:fluid}
%%%%%%%%%%%%%%%%%%%%%%%%%%%%%%%%%%%%%%%%%%%%%%
We define the process describing the behavior of the underlying system as
$\Markov=(\Markov(t), t=0,1,2,\cdots)$, where
\[
\Markov(t) \triangleq \left( (Z_{\sk,1}(t), \cdots, Z_{\sk,Q_{\sk}(t)}(t)), 
(\sk) \in \Pair \right).
\]
We define the norm of $\Markov(t)$ as
\begin{equation}
\| \Markov(t) \| \triangleq \| \vec{Q}(t) \| + \| \vec{W}(t) \|.
\end{equation}
% The evolution of $\Markov$ forms a discrete-time Markov chain, if a 
% scheduling decision is based on the information of the current time 
% slot only. It is clear that $\Markov$ forms a Markov chain with countable
% state space under Q-BP. Let $\Markov^{(x)}$ denote a process $\Markov$ 
% with an initial configuration such that 
Clearly, under Q-BP, the evolution of $\Markov$ forms a discrete-time 
Markov chain with countable state space. Let $\Markov^{(x)}$ denote a 
process $\Markov$ with an initial configuration such that 
\begin{equation}
\label{eq:init_config}
\| \Markov^{(x)}(0) \| = x. 
\end{equation}
% All the processes of $\Markov^{(x)}$ satisfy the properties in the original 
% system $\Markov$.

The following Lemma was derived in \cite{rybko92} for continuous-time 
countable Markov chains, and it follows from more general results in 
\cite{malyshev79} for discrete-time countable Markov chains. 

\high{
\begin{lemma}[Theorem~4 of \cite{andrews04}]
\label{lem:stab_cri}
Suppose there exist an $\epsilon>0$ and a finite integer $T > 0$ such 
that for any sequence of processes $\{\frac {1} {x} \Markov^{(x)}(x T), 
x=1,2,\cdots\}$, we have
\begin{equation}
\label{eq:stab_cri}
\textstyle \limsup_{x \rightarrow \infty} \Expect \left[\frac {1} {x} 
\| \Markov^{(x)} (x T) \| \right] \le 1 - \epsilon.
\end{equation}
Then, the Markov process $\Markov$ is positive recurrent.
\end{lemma}
}

A stability criteria of (\ref{eq:stab_cri}) leads to a fluid limit 
approach \cite{dai95,stolyar95} to the stability problem of queueing 
systems. Hence, we start our analysis by establishing the \emph{fluid 
limit model} as in~\cite{dai95, andrews04}. We define the process 
$\System \triangleq \left(A,F,\hat{F},Q,P,\Pi,\Psi,W,U \right)$, and 
it is clear that a sample path of $\System^{(x)}$ uniquely defines the 
sample path of $\Markov^{(x)}$. Then we extend the definition of $Y=A, 
F, \hat{F}, Q, P, \Pi, \Psi, W$ and $U$ to continuous time domain as 
$Y(t) \triangleq Y(\lfloor t \rfloor)$ for each continuous time $t \ge 0$. 

As in \cite{andrews04}, we extend the definition of $F^{(x)}_s(t)$ to 
the negative interval $t \in [-x,0)$ by assuming that the packets present 
in the initial state $\Markov^{(x)}(0)$ arrived in the past at some of 
the time instants $-(x - 1), -(x - 2), \cdots, 0$, according to their 
delays in the state $\Markov^{(x)}(0)$. By this convention, we have 
$F^{(x)}_s(-x)=0$ for all $s\in\Flow$ and $x$, and $\sum_{s \in S} 
F^{(x)}_s(0) \le x$ for all $x$. 

Then, applying the techniques used in the proof for Theorem~4.1 of \cite{dai95} 
or Lemma~1 of \cite{andrews04}, we can show that with probability one, for any 
sequence of processes $\{\frac {1} {\xn} \System^{(\xn)}(\xn \cdot)\}$, 
where $\{\xn\}$ is a sequence of positive integers with $\xn \rightarrow 
\infty$, there exists a subsequence $\{\xnj\}$ with $\xnj \rightarrow 
\infty$ as $j \rightarrow \infty$ such that the following convergences 
hold \emph{uniformly over compact (u.o.c.)} intervals:
\begin{eqnarray}
&&	\textstyle\frac{1}{\xnj}\int_0^{\xnj t} A^{(\xnj)}_s(\tau)d\tau \rightarrow  \lambda_s t, \label{eq:fluid_a}\\
&&	\textstyle\frac{1}{\xnj}F^{(\xnj)}_s(\xnj t) \rightarrow  f_s(t), \\
&&	\textstyle\frac{1}{\xnj}\hat{F}^{(\xnj)}_\sk(\xnj t) \rightarrow  \hat{f}_\sk(t), \\
&&	\textstyle\frac{1}{\xnj}Q^{(\xnj)}_\sk(\xnj t) \rightarrow  q_\sk(t), \label{eq:fluid_q} \\
&&	\textstyle\frac{1}{\xnj}P^{(\xnj)}_\sk(\xnj t) \rightarrow  p_\sk(t), \\
&&	\textstyle\frac{1}{\xnj}\int_0^{\xnj t} \Pi^{(\xnj)}_\sk(\tau)d\tau \rightarrow  \int_0^t \pi_\sk(\tau) d\tau,\\
&&	\textstyle\frac{1}{\xnj}\int_0^{\xnj t} \Psi^{(\xnj)}_\sk(\tau)d\tau \rightarrow  \int_0^t \psi_\sk(\tau) d\tau. \label{eq:fluid_psi}
\end{eqnarray}
Similarly, the following convergences (which are denoted by ``$\Rightarrow$") 
hold at every continuous point of the limit function:
\begin{eqnarray}
&&	\textstyle\frac{1}{\xnj}W^{(\xnj)}_\sk(\xnj t) \Rightarrow  w_{\sk}(t), \label{eq:fluid_w} \\
&&	\textstyle\frac{1}{\xnj}U^{(\xnj)}_\sk(\xnj t) \Rightarrow  u_{\sk}(t). \label{eq:fluid_u}
\end{eqnarray}
\high{The above convergence properties follow directly from the Arzela-Ascoli 
Theorem and the structure of the model: that the arrival process satisfies the 
SLLN and that the sequence of the (scaled) departure process is uniformly bounded 
and uniformly equicontinuous.} 

Any set of limiting functions $(f, \hat{f}, q, p, \pi, \psi, w, u)$
is called a \emph{fluid limit}. \high{The family of these fluid limits 
is associated with our original stochastic network. The scaled sequences 
$\{\frac {1} {\xn} \System^{(\xn)}(\xn \cdot)\}$ and their limits are 
referred to as a \emph{fluid limit model} \cite{bramson08}.}
Since some of the limiting functions, namely $f_s, \hat{f}_\sk, q_\sk, 
p_\sk$ are Lipschitz continuous in $[0,\infty)$, they are absolutely 
continuous. Hence, at almost all points $t \in [0, \infty)$, the 
derivatives of these limiting functions exist. We call such points 
{\em regular} time. 

We then present the \emph{fluid model equations} of the system as follows. 
\begin{align}
& \textstyle \sum_{s \in \Flow} f_s(0) \le 1, \label{eq:initial} \\
& \textstyle p_\sk(t) = \sum_{i=1}^{k} q_{s,i}(t), \label{eq:pq} \\
& \textstyle p_\sk(t) = f_s(t) - \hat{f}_\sk(t), \label{eq:pf} \\
& \textstyle f_s(t) = f_s(0) + \lambda_s t, \label{eq:ff} \\
% \end{align}
% \begin{align}
& \textstyle u_{\sk}(t) = t - w_{\sk}(t), \label{eq:uw} \\
& \textstyle \psi_\sk(t) \le \pi_\sk(t), \label{eq:pp} \\
& \textstyle \Delta q_\sk(t) = q_\sk(t) - q_{s,k+1}(t), \label{eq:qdfluid} \\
& \textstyle \frac{d}{dt} q_\sk(t) = \left\{
\begin{array}{ll}
\psi_{s,k-1}(t) - \pi_\sk(t), & \text{if}~ q_\sk(t) > 0, \\
(\psi_{s,k-1}(t) - \pi_\sk(t))^+, & \text{otherwise},
\end{array}
\right. \label{eq:dq}
\end{align}
where $(z)^+ \triangleq \max(z,0)$, and we set $\psi_{s,0}=\pi_{s,0}=\lambda_s$.
\high{Fluid model equations can be thought of as belonging to a fluid network 
which is the deterministic equivalence of the original stochastic network. Any 
set of functions satisfying the fluid model equations is called a \emph{fluid 
model solution} of the system. It is easy to check that any fluid limit is a 
fluid model solution.}

It is clear from (\ref{eq:qbp}) that Q-BP will not schedule link-flow-pair 
$(\sk)$ if $Q_\sk(t) - Q_{s,k+1}(t) < 0$. \high{Hence, if link-flow-pair $(\sk)$ 
is scheduled, it must satisfy that $Q_\sk(t) - Q_{s,k+1}(t) \ge 0$. Moreover,  
the length of queue $Q_\sk$ can decrease by at most one within one time slot, 
and the length of queue $Q_{s,k+1}$ can increase by at most one within one 
time slot, due to the assumption of unit link capacity (a similar argument 
also holds with non-unit link rates). This implies that, if 
\begin{equation}
\label{eq:qdcon}
Q_\sk(t) \ge Q_{s,k+1}(t) - 2
\end{equation}
initially holds for all $(\sk)$ at time slot 0, then the inequality holds for 
every time slot $t\ge0$. This further implies that 
\begin{equation}
\label{eq:dqg}\
q_\sk(t) \ge q_{s,k+1}(t), ~\text{i.e.},~\Delta q_\sk(t) \ge 0,
\end{equation}
for all (scaled) time $t \ge 0$, from the convergence of (\ref{eq:fluid_q}).} 
We assume that at time slot 0, all queues on the route of each flow are empty 
except for the first queue, then it follows that (\ref{eq:qdcon}) holds for 
all (scaled) time $t\ge0$, and thus, $\Delta q_\sk(t) \ge 0$ holds for all 
time $t \ge 0$. 

\high{
\noindent \HRule

\emph{Remark:} Note that we make the assumption of empty queues for 
ease of analysis. Even without this assumption, we can show that there 
exists a finite time $T > 0$ such that for all time $t \ge T$, (\ref{eq:dqg}) 
holds for all $(\sk) \in \Pair$. This can be proved by induction. The 
detailed proof can be found in Appendix~\ref{app:lem:pd},
% detailed proof can be found in our online technical report \cite{ji12tr},
but the basic idea is as follows: Consider a flow $\hat{s} \in \Flow$. 
We want to show that there exists a finite time $T_{\hat{s}}>0$ such that
for all time $t \ge T_{\hat{s}}$, (\ref{eq:dqg}) holds for all $(\hat{s},k)$
with $k \in \{1,2,\cdots, \Hop(\hat{s}) \}$.
\begin{enumerate}
\item First, we show that there exists a finite time $T_{\hat{s},1} > 0$ 
such that for all time $t \ge T_{\hat{s},1}$, (\ref{eq:dqg}) holds for 
link-flow-pair $(\hat{s},1)$. Suppose that (\ref{eq:dqg}) does not hold 
for $(\hat{s},1)$. Then Q-BP does not schedule $(\hat{s},1)$, i.e., 
$q_{\hat{s},1}(t)$ does not decrease and $q_{\hat{s},2}(t)$ does not 
increase. On the other hand, due to the exogenous arrivals at the source 
node of flow $\hat{s}$, $q_{\hat{s},1}(t)$ must increase with time. Hence, 
there must exist a finite time $T_{\hat{s},1}$ such that (\ref{eq:dqg}) 
holds for $(\hat{s},1)$ at time $T_{\hat{s},1}$. We can further show that 
(\ref{eq:dqg}) holds for all $t \ge T_{\hat{s},1}$ under Q-BP. This can 
be proved by contradiction.

\item Then, we discuss the induction step: Consider $k \in \{1,2,\cdots,
\Hop(\hat{s})-1\}$. Suppose that for all time $t \ge T_{\hat{s},k} > 0$, 
(\ref{eq:dqg}) holds for $(\hat{s},j)$ and for all $j \in \{1,2,\cdots,k\}$, 
we show that there exists a finite time $T_{\hat{s},k+1} \ge T_{\hat{s},k}$ 
such that for all time $t \ge T_{\hat{s},k+1}$, (\ref{eq:dqg}) holds for 
$(\hat{s},j^{\prime})$ and for all $j^{\prime} \in \{1,2,\cdots,k+1\}$. For 
simplicity, we consider the case for which $k=1$, and the general induction 
step follows similarly. Now, suppose that (\ref{eq:dqg}) does not hold for 
$(\hat{s},2)$, and we prove it by contradiction. Clearly, Q-BP will schedule 
only link-flow-pairs for which (\ref{eq:dqg}) holds (i.e., link-flow-pair 
$(\hat{s},1)$ in this case). Hence, the fluid limit model of the subsystem 
that consists of link-flow-pairs for which (\ref{eq:dqg}) holds must be 
stable, from the throughput-optimality of Q-BP (see Proposition~\ref{pro:qbp}). 
This, in particular, implies that $q_{\hat{s},1}$ is stable, which further 
implies that $q_{\hat{s},2}(t)$ must increase with time, because Q-BP keeps 
forwarding packets from $q_{\hat{s},1}$ to $q_{\hat{s},2}$ while not serving 
$q_{\hat{s},2}$. Hence, there must exist a finite time $T_{\hat{s},2} \ge 
T_{\hat{s},1}$ such that for all time $t \ge T_{\hat{s},2}$, (\ref{eq:dqg}) 
holds for $(\hat{s},2)$. 
\end{enumerate}
Hence, letting $T_{\hat{s}} \triangleq T_{\hat{s},\Hop(\hat{s})}$, we have
that for all time $t \ge T_{\hat{s}}$, (\ref{eq:dqg}) holds for all $(\hat{s},k)$
with $k \in \{1,2,\cdots, \Hop(\hat{s}) \}$. Since the above arguments can 
be applied to any flow $\hat{s} \in \Flow$, we can complete the proof by 
setting $T \triangleq \max_{\hat{s} \in \Flow} T_{\hat{s}}$. 

\noindent \HRule
} 
\subsection{Throughput-Optimality of Q-BP} \label{subsec:qbp}
%%%%%%%%%%%%%%%%%%%%%%%%%%%%%%%%%%%%%%%%%%%%%%%%%%%%%%%%%%%%%
\begin{proposition}
\label{pro:qbp}
Q-BP can support any traffic with arrival rate vector that is 
strictly inside $\Lambda^{*}$.
\end{proposition}

\high{Before giving the proof of Proposition~\ref{pro:qbp}, in the
following lemma, we present a linear relation between cumulative 
queue length $p_\sk(t)$ and waiting time $w_\sk(t)$, which is 
used for proving Proposition~\ref{pro:qbp}. 
\begin{lemma}
\label{lem:pw}
For any fixed $t_\sk > 0$, the two conditions $u_{\sk}(t_\sk) > 0$ 
and $\hat{f}_\sk(t_\sk)$ $>$ $f_s(0)$ are equivalent for every 
link-flow-pair $(\sk) \in \Pair$. Further, if the conditions hold, 
we have
\begin{equation}
\label{eq:pw}
p_\sk(t) = \lambda_s w_{\sk}(t),
\end{equation}
for all $t \ge t_\sk$, with probability one.
\end{lemma}

\begin{figure}[t]
\centering
\epsfig{file=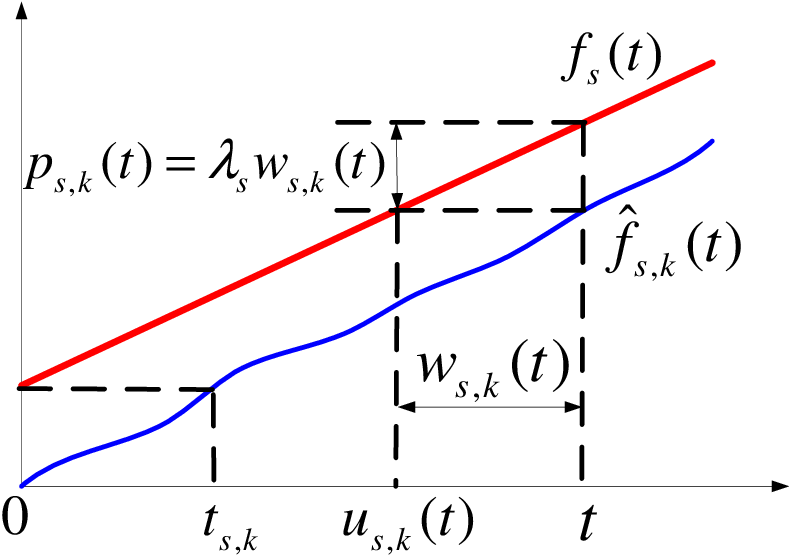,width=0.5\linewidth}
\caption{Linear relation between queue lengths and delays in the fluid limits.}
\label{fig:pw}
\end{figure}

Fig.~\ref{fig:pw} describes the relations between the variables.

\begin{proof}
Since the first part, i.e., that the two conditions are equivalent, is straightforward
from the definition of fluid limits and (\ref{eq:u}), we focus on the second part, 
i.e., if $\hat{f}_\sk(t_\sk) > f_s(0)$, then (\ref{eq:pw}) follow.

Suppose that $\hat{f}_\sk(t_\sk) > f_s(0)$. Then, by the definition of $u_{\sk}(t)$, 
we have $\hat{f}_\sk(t) = f_s(u_{\sk}(t))$, for all $t \ge t_\sk$. From (\ref{eq:pf}), 
(\ref{eq:ff}) and (\ref{eq:uw}), we obtain that
\[
\begin{split}
%\begin{eqnarray}
p_\sk(t) &= f_s(t) - \hat{f}_\sk(t)	\\
		 &= (f_s(0) + \lambda_s t) - (f_s(0) + \lambda_s u_{\sk}(t) ) \\
		 &= \lambda_s \cdot \left( t-u_{\sk}(t) \right)	\\
		 &= \lambda_s w_{\sk}(t).
%\end{eqnarray}
\end{split}
\]
\end{proof}

%Now, we provide the proof for Proposition~\ref{pro:qbp}. 
\begin{proof}[Proof of Proposition~\ref{pro:qbp}]
We prove stability using standard Lyapunov techniques Let 
$V(\vec{q}(t))$ denote the Lyapunov function defined as
\begin{equation}
\textstyle V(\vec{q}(t)) \triangleq \frac 1 2 \sum_{(\sk) \in \Pair}
\left(q_\sk(t)\right)^2.
\end{equation}

From the results of Lemmas~\ref{lem:stab_cri} and \ref{lem:pw}, to show 
positive recurrence, we only need to prove that for any $\zeta>0$, there 
exists a finite time $T_1 > 0$ such that for any fluid limit with 
$\|\vec{q}(0)\| \le 1$, we have 
\begin{equation}
%\label{eq:qt}
\| \vec{q}(t)\| \le \zeta,
\end{equation}
for all time $t \ge T_1$. To show the above, it is sufficient to show 
that for any $\zeta_1 > 0$, there exists $\zeta_2>0$ such that 
$V(\vec{q}(t)) \ge \zeta_1$ implies $\frac{D^+}{dt^+} V(\vec{q}(t)) 
\le -\zeta_2$ for any regular time $t \ge 0$, where $\frac{D^+}{dt^+} 
V(\vec{q}(t)) = \lim_{\delta \downarrow 0} \frac {V(\vec{q}(t+\delta)) 
- V(\vec{q}(t))} {\delta}$.

Suppose $\vlambda$ is strictly inside $\Lambda^{*}$, then there exists 
a vector $\vphi \in Co(\Matching_{\Pair})$ such that $\vlambda < \vphi$, 
i.e., $\lambda_s < \phi_{\sk}$ for all $(\sk) \in \Pair$. Since $\vec{q}(t)$ 
is differentiable, then for any regular time $t \ge 0$, we can obtain the 
derivative of $V(\vec{q}(t))$ as
\begin{equation}
\label{eq:dol}
\begin{split}
%
%\begin{eqnarray}
\textstyle \frac{D^+}{dt^+} & V(\vec{q}(t)) \\
\stackrel{(a)}=& \textstyle \sum_{(\sk) \in \Pair} q_\sk(t) \cdot \left(\psi_{s,k-1}(t) - \pi_\sk(t)\right)  \\
\stackrel{(b)}\le& \textstyle \sum_{(\sk) \in \Pair} q_\sk(t) \cdot \left(\pi_{s,k-1}(t) - \pi_\sk(t)\right)  \\
%=& \sum_{s} q_{s,1}(t) \cdot \lambda_s - \sum_{(\sk) \in \Pair} \Delta q_\sk(t) \cdot \pi_\sk(t) \\
%=& \sum_{s} \left(\sum_{k=1}^{\Hop(s)} \Delta q_\sk(t) \right) \cdot \lambda_s - \sum_{(\sk) \in \Pair} \Delta q_\sk(t) \cdot \pi_\sk(t) \\
=& \textstyle \sum_{(\sk) \in \Pair} \Delta q_\sk(t) \cdot \lambda_s \\
~& \textstyle - \sum_{(\sk) \in \Pair} \Delta q_\sk(t) \cdot \pi_\sk(t) \\
=& \textstyle \sum_{(\sk) \in \Pair} \Delta q_\sk(t) \cdot \left( \lambda_s - \phi_\sk \right) \\
~&+ \textstyle \sum_{(\sk) \in \Pair} \Delta q_\sk(t) \cdot (\phi_\sk-\pi_\sk(t)), \\
%\end{eqnarray}
\end{split}
\end{equation}
where (a) and (b) are from (\ref{eq:dq}) and (\ref{eq:pp}), respectively.

Note that $q_\sk(t) \le \Hop^{\max} \max_{(\rj) \in \Pair} \Delta q_\rj(t)$, 
for any $(\sk) \in \Pair$. Hence, we have $V(\vec{q}(t)) \le \frac {1} {2} 
|S| \Hop^{\max} (\Hop^{\max} \max_{(\rj) \in \Pair} \Delta q_\rj(t))^2$. 
Let us choose $\zeta_3 = \sqrt{\frac {2\zeta_1} {|S|(\Hop^{\max})^3}}$, 
then $V(\vec{q}(t)) \ge \zeta_1$ implies $\max_{(\rj) \in \Pair} \Delta 
q_\rj(t) \ge \zeta_3$. Since $\vlambda < \vphi(t)$ and $\Delta q_\sk(t) 
\ge 0$ for all $(\sk) \in \Pair$, then in the final result of (\ref{eq:dol}), 
we can conclude that the first term is bounded as follows:
\[
\begin{split}
\textstyle \sum_{(\sk) \in \Pair} \Delta q_\sk(t) \cdot \left( \lambda_s - \phi_\sk \right)
& \textstyle \le -\zeta_3 \min_{s,k}(\phi_\sk-\lambda_s) \\
& \triangleq -\zeta_2 < 0,
\end{split}
\]
and that the second term becomes non-positive due to the following. 
Since Q-BP chooses schedules that maximize the queue differential 
weight sum (\ref{eq:qbp}), then we have that 
\[
\textstyle \vpi(t) \in \argmax_{\vphi \in 
Co(\Matching_{\Pair})} \sum_{(\sk) \in \Pair} \Delta q_\sk(t) \cdot \phi_\sk,
\]
which implies that 
\[
\textstyle \sum_{(\sk) \in \Pair} \Delta q_\sk(t) 
\cdot \phi_\sk \le \sum_{(\sk) \in \Pair} \Delta q_\sk(t) 
\cdot \pi_\sk(t),
\]
for all $\vphi \in Co(\Matching_{\Pair})$. Therefore, this shows that 
$V(\vec{q}(t)) \ge \zeta_1$ implies $\frac{D^+}{dt^+} V(\vec{q}(t)) 
\le -\zeta_2$. Then, it immediately follows that for any $\zeta>0$,
there exists a finite time $T_1 > 0$ such that for any fluid limit 
with $\|\vec{q}(0)\| \le 1$, we have $\| \vec{q}(t)\| \le \zeta$ for 
any time $t \ge T_1$. Also, we have
\begin{equation}
\label{eq:plezeta}
p_\sk(t) \le \| \vec{q}(t) \| \le \zeta,
\end{equation}
for all $(s,k) \in \Pair$. Let us choose $T_1$ large enough, then it 
follows from (\ref{eq:initial}), (\ref{eq:pf}) and (\ref{eq:plezeta}) that 
\[
\hat{f}_\sk(T_1) = f_s(T_1) - p_\sk(T_1) > f_s(0),
\]
for all $(s,k) \in \Pair$ and for any time $t \ge T_1$. Hence, we have 
(\ref{eq:pw}) from Lemma~\ref{lem:pw}, and thus, we have 
\[
\begin{split}
\|\vec{q}(t)\| + \| \vec{w} \| 
&\stackrel{(a)}\le \|\vec{q}(t)\| + \frac {1} {\min_s \lambda_s} \| \vec{p}(t) \| \\
%&\stackrel{(b)}\le \left(1 + \frac {|S|\Hop^{\max}} {\min_s \lambda_s} \right) \| \vec{q}(t) \| \\
&\stackrel{(b)}\le \left(1 + \frac {|S|\Hop^{\max}} {\min_s \lambda_s} \right) \zeta \\
&\triangleq \epsilon_1,
\end{split}
\]
where (a) and (b) are from (\ref{eq:pw}) and (\ref{eq:plezeta}), respectively. 
We can make $\epsilon_1$ arbitrarily small by choosing small enough $\zeta$.

Now, consider any fixed sequence of processes $\{ \frac {1} {x} \Markov^{(x)}(xt), 
x=1,2,\cdots\}$ (for simplicity also denoted by $\{x\}$). Hence, for any fixed 
$\epsilon_1>0$, we can always choose a large enough integer $T>0$ such that for 
any subsequence $\{\xn\}$ of $\{x\}$, there exists a further (sub)subsequence 
$\{\xnj\}$ such that
\[
\textstyle \lim_{j \rightarrow \infty} \frac {1} {\xnj} \| \Markov^{(\xnj)} (\xnj T) \| 
= \|\vec{q}(T)\|+\|\vec{w}(T)\| \le \epsilon_1
\]
almost surely. This in turn implies (for small enough $\epsilon_1$) that
\begin{equation}
\label{eq:mc_conv}
\textstyle \limsup_{x \rightarrow \infty} \frac {1} {x} \| \Markov^{(x)} (x T) \| 
\le \epsilon_1 \triangleq 1 - \epsilon < 1
\end{equation}
almost surely. This is because there must exist a subsequence of $\{ x \}$ that 
converges to the same limit as $\limsup_{x \rightarrow \infty} \frac {1} {x} \| 
\Markov^{(x)} (x T) \|$.

One can readily show that the sequence $\{ \frac {1} {x} \| \Markov^{(x)} (x T) \|, 
x=1,2,\cdots \}$ is uniformly integrable using standard techniques by invoking 
the Dominated Convergence Theorem and so the details are omitted here. Then, the 
almost sure convergence in (\ref{eq:mc_conv}) along with uniform integrability 
implies the following convergence in the mean:
\[
\textstyle \limsup_{x \rightarrow \infty} \Expect [ \frac {1} {x} \| 
\Markov^{(x)} (x T) \| ] \le 1 - \epsilon.
\]

Since the above convergence holds for any sequence of processes 
$\{\frac {1} {x} \| \Markov^{(x)} (x \cdot), x=1,2,\cdots\}$, the 
condition of (\ref{eq:stab_cri}) in Lemma~\ref{lem:stab_cri} is 
satisfied. This completes the proof.
\end{proof}
}
%
%
%

%%%%%%%%%%%%%%%%%%%%%%%%%%%%%%%%%%%%%%%%%%%%%%%%%%%%%%%%%%%%%
\section{Delay-based Back-Pressure Algorithm} \label{sec:dbp}
%%%%%%%%%%%%%%%%%%%%%%%%%%%%%%%%%%%%%%%%%%%%%%%%%%%%%%%%%%%%%

%%%%%%%%%%%%%%%%%%%%%%%%%%%%%%%%%%%%%%%%%%%%%%%%%%%%%
\subsection{Algorithm Description} \label{subsec:alg}
%%%%%%%%%%%%%%%%%%%%%%%%%%%%%%%%%%%%%%%%%%%%%%%%%%%%%
\high{In this section, we develop the Delay-based Back-Pressure (D-BP) policy, 
and in Section~\ref{subsec:dbp}, we prove that it is throughput optimal. 
A similar delay-based approach has appeared first in \cite{andrews04} for 
single-hop networks.} However, as mentioned earlier, when packets travel multiple 
hops before leaving the system, the analytical approach in \cite{andrews04} (i.e., 
using HOL delay in the queue as the metric) cannot capture queueing dynamics of 
multihop traffic and the resultant solutions cannot guarantee the linear relation. 
We will carefully design link weights using a new delay metric, and re-establish 
the linear relation between queue lengths and delays in the fluid limits for 
multihop traffic. 

Recall that $W_{\sk}(t)$ denotes the sojourn time of the HOL packet of queue 
$Q_{\sk}(t)$ in the network, where the time is measured from the time when the 
packet arrives in the network. We define the delay metric $\hat{W}_\sk(t)$ as
\begin{equation} 
\label{eq:linkd}
\hat{W}_\sk(t) \triangleq W_{\sk}(t) - W_{s,k-1}(t),
\end{equation}
and also define \emph{delay differential} as
\begin{equation} 
\label{eq:dd}
\Delta \hat{W}_\sk(t) \triangleq \hat{W}_\sk(t) - \hat{W}_{s,k+1}(t).
\end{equation}
The relations between these delay metrics are illustrated in Fig.~\ref{fig:weight}.
We specify the back-pressure algorithm with the new delay metric as follows. 

\noindent {\bf Delay-based Back-Pressure (D-BP) algorithm:}
\begin{equation}
\label{eq:dbp}
\textstyle \vec{M}^* \in \argmax_{\vec{M} \in \Matching_\Pair} \sum_{(\sk) \in \Pair} \Delta \hat{W}_\sk(t) \cdot M_\sk.
\end{equation}

\noindent D-BP computes the weight of $(\sk)$ as the delay differential $\Delta 
\hat{W}_\sk(t)$ and solves the MaxWeight problem, i.e., finds a set of non-interfering 
link-flow-pairs that maximizes weight sum. Ties can be broken arbitrarily if there 
is more than one schedule that has the largest weight sum. An intuitive interpretation 
of the new delay metric $\hat{W}_\sk(t)$ is as follows. Note that the queue length 
$Q_\sk(t)$ is roughly the number of packets arriving at the source node of flow $s$ 
during the time slots between $[U_\sk(t),U_\sk(t)+\hat{W}_\sk(t))$, and from the SLLN,
$Q_\sk(t)$ is on the order of $\lambda_s \hat{W}_\sk(t)$ when $\hat{W}_\sk(t)$ is large. 
Hence, a large $\hat{W}_\sk(t)$ implies a large queue length $Q_\sk(t)$, and similarly, 
a large delay differential $\Delta \hat{W}_\sk(t)$ implies a large queue length 
differential $\Delta Q_\sk(t)$. Therefore, being favorable to the delay weight sum 
in (\ref{eq:dbp}) is in some sense ``equivalent" to being favorable to the queue 
length weight sum in (\ref{eq:qbp}) as Q-BP. We later formally establish the linear 
relation between the fluid limits of queue lengths and delays in Section~\ref{subsec:dbp}. 

\high{
We highlight here that the last packet problem can be solved by the D-BP 
scheme using our proposed delay metric. Let us focus on the source nodes 
first. Suppose that at the source node of flow $s$, there are a finite 
number of packets waiting to be transmitted and there are no further packet 
arrivals. From the definition of (\ref{eq:linkd}) and the fact that $W_{s,0}(t) 
= 0$, we have $\hat{W}_{s,1}(t) = W_{s,1}(t)$. If some of the packets are 
stuck at the source node, the delay metric $\hat{W}_{s,1}(t)$ keeps increasing 
with time. On the other hand, $\hat{W}_{s,2}(t) = W_{s,2}(t) - W_{s,1}(t)$ 
is equal to the inter-arrival time between two packets and does not increase 
with time, in particular because some packets at the source node are not served. 
Hence, the delay differential $\Delta \hat{W}_{s,1}(t) = \hat{W}_{s,1}(t) - 
\hat{W}_{s,2}(t)$ also increases with time. This implies that under DBP, the 
increasing delay will eventually ``push" all the packets that are waiting at 
the source node to the second-hop link. After all the packets leave the source 
node, we can observe similar procedure at the transmitting node of the second-hop 
link: since $Q_{s,1}(t) = 0$ and $W_{s,1}(t) = 0$, we have $\hat{W}_{s,2}(t) = 
W_{s,2}(t)$. Repeating the same argument, we can conclude that all the packets 
will ultimately be ``pushed" to the destination node of flow $s$.
}

\begin{figure}[t]
\centering
\epsfig{file=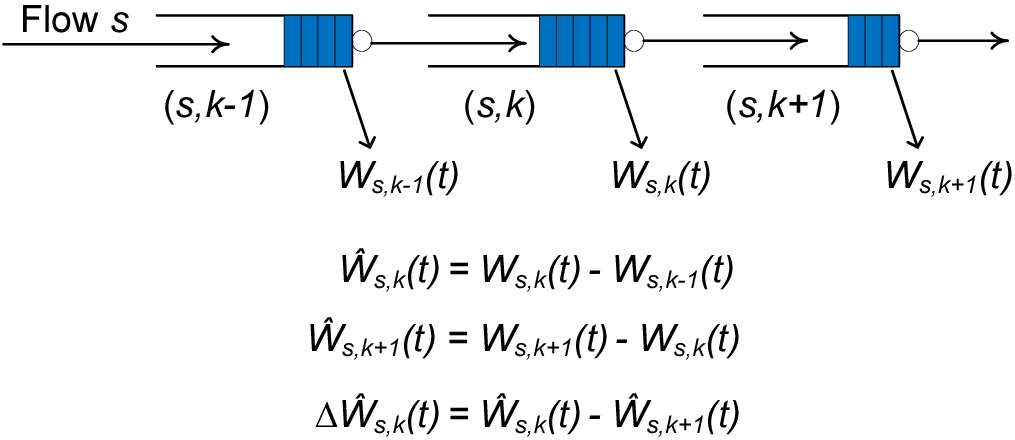,width=0.5\linewidth}
\caption{Delay differentials using new delay metric.}
\label{fig:weight}
\end{figure}

\high{Recall that $U_{\sk}(t)$ denotes the time when the HOL packet of 
$Q_\sk$ arrives in the network (or the source node, rather than the 
current node). We let $U^{\prime}_{\sk}(t)$ denote the time when the 
packet that arrives (in the network or the source node) immediately 
after the HOL packet of $Q_\sk$ arrives in the network. Let $B_{\sk}(t) 
\triangleq U^{\prime}_{\sk}(t) - U_{\sk}(t)$ denote the inter-arrival 
time between the HOL packet of $Q_\sk$ and the packet that arrives 
immediately after it. Clearly, D-BP will not schedule link-flow-pair 
$(\sk)$ if 
\[
\hat{W}_\sk(t) - \hat{W}_{s,k+1} (t) < 0.
\]
Hence, if link-flow-pair $(\sk)$ is scheduled, it must satisfy 
$\hat{W}_\sk(t) - \hat{W}_{s,k+1} (t) \ge 0$. Moreover, the delay 
$\hat{W}_\sk(t)$ can decrease by at most $B_{\sk}(t)$ within one 
time slot, and the delay $\hat{W}_{s,k+1} (t)$ can increase by at 
most $B_{\sk}(t)$ within one time slot, due to the assumption of 
unit link capacity (a similar argument also holds with non-unit 
link rates). Therefore, if inequality
\begin{equation}
\label{eq:wdcon}
\hat{W}_\sk(t) \ge \hat{W}_{s,k+1}(t) - 2 B_{\sk}(t)
\end{equation} 
initially holds for all $(\sk)$ at time slot 0, then the inequality holds 
for all time slot $t\ge0$. This further leads to
\begin{equation}
\label{eq:fluidwcon}
\hat{w}_\sk(t) \ge \hat{w}_{s,k+1}(t),~\text{i.e.,}~\Delta \hat{w}_\sk(t) \ge 0, 
\end{equation}
for all (scaled) time $t\ge0$, in the fluid limits, from the convergence
of (\ref{eq:fluid_w})} and that $\frac{1} {\xnj} B^{(\xnj)}_{\sk}(\xnj t) 
\rightarrow$ $0$, as $\xnj \rightarrow \infty$ (otherwise we will arrive a 
contradiction with the assumption on the arrival process, i.e., it satisfies 
the Strong Law of Large Numbers). Recall that we assume that all queues on 
each route are empty at time slot 0, except for the first queue, then 
(\ref{eq:wdcon}) and (\ref{eq:fluidwcon}) follow.

%%%%%%%%%%%%%%%%%%%%%%%%%%%%%%%%%%%%%%%%%%%%%%%%%%%%%%%%%%%%%%%%%%
\subsection{Throughput-Optimality} \label{subsec:dbp}
%%%%%%%%%%%%%%%%%%%%%%%%%%%%%%%%%%%%%%%%%%%%%%%%%%%%%%%%%%%%%%%%%%
\high{The following lemma provides the linear relation between queue 
lengths and delays in the fluid limits. 

\begin{lemma}
\label{lem:qw}
For any fixed $t_\sk > 0$, if $\hat{f}_\sk(t_\sk)$ $>$ $f_s(0)$
for every link-flow-pair $(\sk) \in \Pair$, then we have
\begin{equation}
\label{eq:qw}
q_\sk(t) = \lambda_s \hat{w}_\sk(t),
\end{equation}
for all $t \ge t_\sk$, with probability one.
\end{lemma}

\begin{proof}
It follows immediately from Lemma~\ref{lem:pw}. 
\end{proof}
}

We emphasize the importance of (\ref{eq:qw}). Lemma~\ref{lem:qw} implies that 
after a finite time (i.e., $\max_{(\sk) \in \Pair} t_\sk$), the queue lengths 
are $\lambda_s$ times delays in the fluid limit model. Then the schedules of 
D-BP are very similar to those of Q-BP, which implies that D-BP achieves the 
optimal throughput region $\Lambda^*$. In the following, \high{we show that 
the condition of Lemma~\ref{lem:qw} indeed holds, i.e.,} such a finite time 
exists.

\begin{lemma}
\label{lem:init}
Consider a system under the D-BP policy. Then for $\vlambda$ strictly inside 
$\Lambda^{*}$, there exists a finite time $T > 0$ such that the fluid limits 
satisfy the following property with probability one, 
\begin{equation}
\hat{f}_\sk(T) > f_s(0),
\end{equation}
for all link-flow-pairs $(\sk) \in \Pair$.
\end{lemma}

We can prove Lemma~\ref{lem:init} by induction following the techniques 
described in Lemma~7 of \cite{andrews04}. The formal proof is provided in 
Appendix~\ref{app:lem:init}. We next outline an informal discussion, which 
highlights the main idea of the proof. First, we consider the base case. D-BP 
chooses one of the feasible schedules in $\Matching_\Pair$ (we omit the term 
``feasible" in the following, whenever there is no confusion) at each time slot. 
Each schedule receives a fraction of the total time and there must exist a 
schedule that receives at least $\frac {1} {|\Matching_\Pair|}$ fraction of 
the total time. Thus, after a large enough time $T_1 > 0$, there must exist a 
schedule $\vM^*$ that is chosen for at least $\frac {T_1} {|\Matching_\Pair|}$ 
amount of time. The number of initial packets of $\vM^*$ is bounded from 
(\ref{eq:initial}), thus, for a large enough $T_1$, all initial ``fluid"
of at least one link-flow-pair of $\vM^*$ must be completely served, i.e., 
$\hat{f}_\sk(T_1) > f_s(0)$, for at least one $(\sk)$ with $M^*_{\sk} = 1$.

Next, we consider the inductive step. Suppose there exists a $T_l > 0$, 
such that for at least one subset $S_l \subset \Pair$ of cardinality $l$, 
we have 
\begin{equation}
\label{eq:ind1}
\hat{f}_\sk(T_l) > f_s(0),
\end{equation}
for all $(\sk) \in S_l$.
Then there exists $T_{l+1} \ge T_l$ such that
\begin{equation}
\label{eq:ind2}
\hat{f}_\sk(T_{l+1}) > f_s(0),
\end{equation}
holds for all link-flow-pairs $(\sk)$ within at least one subset $S_{l+1} \subset 
\Pair$ of cardinality $l+1$. \high{Since flows travel hop-by-hop, packets that have been 
served by one link must have been served by the link at the previous hop (of the 
flow that the packets belong to). Hence, if $(s,k) \in S_l$, we must have $(s,k-1) 
\in S_l$. Repeating the argument, if $(s,k) \in S_l$, we have $(s,i) \in S_l$ for 
$1 \le i \le k$.} Let 
\begin{equation}
\begin{split}
\label{eq:s1s}
S^*_l \triangleq \{(\rj) ~|~ & (\rj) \notin S_l, (\rj-1) \in S_l, ~\text{for}~ j>1; \\
&~\text{or}~ (\rj) \notin S_l, ~\text{for}~ j=1 \}
\end{split}
\end{equation}
denote the set of link-flow-pairs $(\rj)$ such that $(\rj) \in \Pair \backslash S_l$ 
is the closest hop to the source of $r$. To avoid unnecessary complications, we discuss 
the induction step for $l=1$. The generalization for $l>1$ is straightforward. We show 
that for given $S_1$ and $T_1$, there exists a finite time $T_2 \ge T_1$ such that 
(\ref{eq:ind2}) with $T_2$ holds for at least two different link-flow-pairs.

Let $(\hsk)$ denote the link-flow-pair that satisfies (\ref{eq:ind1}) with $T_1$. 
Since $(\hsk) \in S_l$ implies $(\hat{s},i) \in S_l$ for all $1 \le i \le \hat{k}$, 
we must have $\hat{k} = 1$ and $S_1 = \{(\hat{s}, 1)\}$. From (\ref{eq:s1s}), we have 
that 
\begin{equation}
\label{eq:s1ss}
S_1^* = \{ (r,1)~|~ r \in \Flow \backslash \{ \hat{s} \} \} \cup N_{\hat{s}}, 
\end{equation}
where $N_{\hat{s}} = \{(\hat{s},2)\}$ if $\Hop(\hat{s}) > 1$, and $N_{\hat{s}} = \emptyset$ 
if $\Hop(\hat{s}) = 1$. We discuss only the case that $\Hop(\hat{s}) > 1$, and the other 
case can be easily shown following the same line of analysis. Now suppose that 
\begin{equation}
\label{eq:as}
\hat{f}_{\rj}(t) \le f_r(0), ~\text{for all}~ (\rj) \in \Pair \backslash S_1, 
~\text{and all}~ t \ge 0,
\end{equation}
i.e., for all the link-flow-pairs except those of $S_1$, the total amount of
service up to time $t$ is no greater than the amount of the initial fluid for 
all $t\ge0$. We show that this assumption leads to a contradiction, which 
completes the induction step. 

From the base case and Lemma~\ref{lem:qw}, we have $q_{\hat{s},1}(t)$ $=$ 
$\lambda_{\hat{s}} \hat{w}_{\hat{s},1}(t)$ for all $t \ge T_1$. We view the 
subset of link-flow-pairs $S_1$ as a generalized system, and consider the 
time slots when there is at least one packet transmission from the outside 
of $S_1$, i.e., $(\rj) \in \Pair \backslash S_1$. For each such time slot, 
we say that the time slot is \emph{unavailable} to $S_1$.

\begin{enumerate}

\item The number of such unavailable time slots is bounded from the above by 
$\xnj$, since at every such time slot, at least one initial packet will be 
transmitted and the total number of initial packets is bounded by $\|\vec{Q}(0)\|
\le \xnj$ from (\ref{eq:init_config}). Hence, the amount of (scaled) time 
unavailable to $S_1$ is bounded by $\|\vec{q}(0)\| \le 1$.

\item Since the amount of (scaled) time unavailable to $S_1$ is bounded, there exists 
a sufficiently large $t \ge T_1$ such that the fraction of time that is given to $(\rj) 
\in \Pair \backslash S_1$ is negligible, and we must have $\hat{w}_{\hat{r}, \hat{j}}(t) 
= \Theta(1)$\footnote{We use the standard order notation: $g(n)=o(f(n))$ implies 
$\lim_{n \rightarrow \infty} (g(n)/f(n)) = 0$; and $g(n)=\Theta(f(n))$ implies 
$c_1 \le \lim_{n \rightarrow \infty} (g(n)/f(n)) \le c_2$ for some constants $c_1$ 
and $c_2$.} and $\Delta \hat{w}_{\hat{r}, \hat{j}}(t) = \Theta(1)$ for 
$(\hat{r},\hat{j}) \in \Pair \backslash (S_1 \cup S_1^*)$. 

\item Then, we can restrict our focus on the generalized system $S_1$ to time
$t \ge T_1$, and ignore the time that is unavailable to $S_1$. Then Q-BP and 
D-BP are in some sense ``equivalent" in the generalized system $S_1$ for $t \ge T_1$ 
with the following properties: First, Q-BP will stabilize the system if the arrival 
rate vector is strictly inside $\Lambda^{*}$. Second, since the linear relation (\ref{eq:qw}) holds 
for all link-flow-pairs in $S_1$ from Lemma~\ref{lem:qw}, D-BP will schedule links 
similar to Q-BP and also stabilizes the generalized system $S_1$. 

\item Now let us focus on $S_1^*$. Link-flow-pairs in $S_1^*$ must have some 
initial fluid at $t \ge T_1$ because $S_1 \cap S_1^* = \emptyset$. On the other 
hand, the generalized network $S_1$ is stable. This implies that the delay metrics 
of link-flow-pairs in $S_1^*$ should increase on the same order as we increase $t$, 
i.e., $\hat{w}_{r^*,j^*}(t) = \Theta(t)$ for $(r^*,j^*) \in S_1^*$. Then we have 
$\Delta \hat{w}_{r^*,j^*}(t) = \Theta(t)$, since $\hat{w}_{r^*,j^*+1}(t) = \Theta(1)$ 
from $(r^*,j^*+1) \in \Pair \backslash (S_1 \cup S_1^*)$ and 2). Since the delay 
differentials $\Delta \hat{w}_{\sk}(t)$ for all $(\sk) \in S_1$ and $\Delta \hat{w}_{\hat{r}, 
\hat{j}}(t)$ for all $(\hat{r}, \hat{j}) \in \Pair \backslash (S_1 \cup S_1^*)$ are 
bounded above from stability of $S_1$ and 2), respectively, D-BP will choose link-flow-pairs 
in the set of $S_1^*$ for most of time for a sufficiently large $t$. This implies that 
the amount of time unavailable to $S_1$ is $\Theta(t)$, which contradicts with our 
previous statement in 1) that the fraction of time that is given to $(\rj) \in \Pair 
\backslash S_1$ is negligible.

\end{enumerate}

We provide the detailed proof of Lemma~\ref{lem:init} in Appendix~\ref{app:lem:init}. 

We then present throughput-optimality of D-BP in the following proposition.
\begin{proposition}
\label{pro:dbp}
D-BP can support any traffic with arrival rate vector that is strictly inside $\Lambda^{*}$.
\end{proposition}
\begin{proof}
We show the stability using fluid limits and standard Lyapunov techniques. 
From Lemmas~\ref{lem:qw} and \ref{lem:init}, we obtain the key property
for proving throughput-optimality of D-BP in Eq.~(\ref{eq:qw}), i.e., 
after a finite time, there is a linear relation between queue lengths
and delays in the fluid limit model. We start with the following quadratic-form 
Lyapunov function, 
\begin{equation}
\textstyle V(\vec{q}(t)) \triangleq \frac 1 2 \sum_{(\sk) \in \Pair} 
\frac {\left(q_\sk(t)\right)^2} {\lambda_s}.
\end{equation}
Following the line of analysis in the proof for Proposition~\ref{pro:qbp}, 
we can show that for any $\zeta_1 > 0$, there exist $\zeta_2>0$ and a finite 
time $T>0$ such that $V(\vec{q}(t)) \ge \zeta_1$ implies $\frac{D^+}{dt^+} 
V(\vec{q}(t)) \le -\zeta_2$ for any regular time $t \ge T$, if the underlying 
scheduler maximizes $\sum_{s,k} \frac{\Delta q_{s,k}(t)}{\lambda_s} \cdot 
\pi_{s,k}(t)$. Then, by applying the linear relation (\ref{eq:qw}), we can 
see that D-BP indeed satisfies such a condition, and obtain the results. We 
omit the detailed proof since it mirrors the derivations in Proposition~\ref{pro:qbp}.
\end{proof}

% From Lemmas~\ref{lem:qw} and \ref{lem:init}, we observe that after some finite 
% period of time, there is a linear relation between $q_\sk(t)$ and $\hat{w}_\sk(t)$, 
% and thus Q-BP and D-BP are in some sense ``indistinguishable" in the fluid limit 
% model. Hence, as in \cite{andrews04}, we can define a class of MaxWeight-type 
% scheduling policies that compute weight as $\gamma_\sk \left (\hat{W}_\sk(t) + 
% \eta_\sk Q_\sk(t) \right )^\zeta$ for some $\gamma_\sk$, $\eta_\sk$ and $\zeta$. 
% Similarly, we can show that this class of scheduling policies is also throughput-optimal. 
% This generalization allows more flexible controls of queue lengths and delays by 
% introducing additional parameters $\gamma_\sk$, $\eta_\sk$ and $\zeta$.

%%%%%%%%%%%%%%%%%%%%%%%%%%%%%%%%%%%%%%%%%%%%%%%%%%
\section{Greedy Algorithms} \label{sec:greedy}
%%%%%%%%%%%%%%%%%%%%%%%%%%%%%%%%%%%%%%%%%%%%%%%%%%
It is well known that the schemes (e.g., Q-BP and D-BP) based on the back-pressure 
techniques are complex to implement because they involve computing a MaxWeight 
component, which in general is NP-hard \cite{sharma06}. Hence, although D-BP operates 
efficiently and achieves the optimal throughput region, it could be difficult to 
implement in practice. Therefore, we are interested in simpler approximations of 
D-BP that can achieve a guaranteed fraction of the optimal performance. The 
Delay-based Greedy Maximal Scheduling (D-GMS) algorithm is a good candidate 
approximation algorithm. A Greedy Maximal Scheduling (GMS) algorithm 
\cite{dimakis06,lin06, joo09b,leconte09} (which is also known as Longest Queue 
First (LQF)) operates (in the scenarios with single-hop traffic) as follows: at 
each time slot $t$, starts with an empty schedule; first picks a link $l$ with 
the maximum weight (e.g., queue length or delay); adds $l$ into the schedule, 
and disables other links that interfere with $l$; next picks a link $l^\prime$ with 
the maximum weight from the remaining set of links, adds $l^\prime$ into the schedule, 
and disables other links that interfere with $l^\prime$; and continues this process 
until all links are either chosen or disabled. All chosen links will be scheduled 
during time slot $t$. Note that any schedule obtained by GMS is maximal. 

GMS has been extensively studied due to its low complexity \cite{lin06}, distributed 
implementations \cite{hoepman04} (or distributed approximations \cite{joo11}) and 
empirically observed good performance \cite{joo09c}. It was first shown in \cite{dimakis06} 
that GMS is throughput-optimal in networks where the so-called \emph{local pooling} 
condition is satisfied. The authors of \cite{joo09,joo09b} generalize the idea of 
\emph{local pooling} to $\sigma$-\emph{local pooling}, where $\sigma$ is a topological 
notion depending on the underlying network topology and is called the \emph{local pooling
factor}. There, the authors show that GMS can achieve a $\sigma$-fraction of the optimal 
throughput region. On the other hand, in \cite{brzezinski08,zussman08}, the \emph{local 
pooling} condition is generalized to the scenarios with multihop traffic, i.e., GMS is 
throughput-optimal in networks where the \emph{multihop local-pooling} condition is 
satisfied. Next, we will discuss the performance limits of D-GMS.

\begin{algorithm}[t]
\caption{Greedy Maximal Scheduling (GMS) Algorithm} \label{alg:gms}
\begin{algorithmic}[1]
\Procedure{GMS}{$\Pair,x$}
\State $M \gets \emptyset$
\State $\Pair^\prime \gets \Pair$
\While{$\Pair^\prime \neq \emptyset$}
\State pick a link-flow-pair $(\sk)$ with maximum weight: $x(\sk)=\max_{(\rj) \in \Pair^\prime} x(\rj)$
\State $M \gets M \cup \{(\sk)\}$
\State $\Pair^\prime \gets \Pair^\prime \backslash I(\sk)$
\EndWhile
\EndProcedure
\end{algorithmic}
\end{algorithm}

% \begin{algorithm}[t]
% \caption{Delay-based Greedy Maximal Scheduling (D-GMS) Algorithm} \label{alg:dgms}
% \begin{algorithmic}[1]
% \Procedure{D-GMS}{$\Pair,x$}
% \State At each time slot $t$:
% \For{$(\sk) \in \Pair$}
% \State $x(\sk) \gets \Delta \hat{W}_\sk(t)$
% \EndFor
% \State GMS($\Pair,x$)
% \EndProcedure
% \end{algorithmic}
% \end{algorithm}
%
% \begin{algorithm}[t]
% \caption{Queue-length-based Greedy Maximal Scheduling (Q-GMS) Algorithm} \label{alg:qgms}
% \begin{algorithmic}[1]
% \Procedure{Q-GMS}{$\Pair,x$}
% \State At each time slot $t$:
% \For{$(\sk) \in \Pair$}
% \State $x(\sk) \gets \Delta Q_\sk(t)$
% \EndFor
% \State GMS($\Pair,x$)
% \EndProcedure
% \end{algorithmic}
% \end{algorithm}

To generalize the GMS algorithm to settings with multihop traffic, we consider 
link-flow-pairs. We let $x(\sk)$ denote the weight of link-flow-pair $(\sk) \in \Pair$, 
and conclude the procedure of GMS in Algorithm~\ref{alg:gms}. We then describe the 
operations of D-GMS and its queue-length-based counterpart (called Q-GMS) in the
following.

\noindent {\bf Delay-based Greedy Maximal Scheduling (D-GMS) Algorithm:}
At each time slot $t$, the algorithm sets the weight of each link-flow-pair
to the delay differential, i.e.,
\begin{equation}
\label{eq:dgms}
x(\sk) \gets \Delta \hat{W}_\sk(t), ~\text{for all}~ (\sk) \in \Pair,
\end{equation}
and finds its schedule in decreasing order of weight conforming to the
underlying interference constraints, by applying Algorithm~\ref{alg:gms}.

\noindent {\bf Queue-length-based Greedy Maximal Scheduling (Q-GMS) Algorithm:}
At each time slot $t$, the algorithm sets the weight of each link-flow-pair
to the queue-length differential, i.e.,
\begin{equation}
\label{eq:qgms}
x(\sk) \gets \Delta Q_\sk(t), ~\text{for all}~ (\sk) \in \Pair,
\end{equation}
and finds its schedule by applying Algorithm~\ref{alg:gms}.
% Ties can be broken arbitrarily if there is more than one 
% link-flow-pair that has the maximum weight. 

We characterize the throughput performance of D-GMS in the following proposition. 

\begin{proposition}
The achievable throughput region of D-GMS is no smaller than that of Q-GMS.
\end{proposition}

We omit the proof here, since it follows the similar line of analysis for D-BP
to establish the linear relation between queue lengths and delays in the fluid
limits, and the result can then be obtained by applying the techniques used in 
\cite{brzezinski08,zussman08}.
%
%
%

%%%%%%%%%%%%%%%%%%%%%%%%%%%%%%%%%%%%%%%%%%%%%%%%%%
\section{Numerical Results} \label{sec:simulation}
%%%%%%%%%%%%%%%%%%%%%%%%%%%%%%%%%%%%%%%%%%%%%%%%%%
In this section, we first highlight the \emph{last packet problem} for the
queue-length-based back-pressure algorithm. The last packet problem implies 
that flows that lack packet arrivals at subsequent time may experience 
excessive delays under Q-BP, which is later confirmed in the simulations. 
Then, we compare throughput and delay performance of Q-BP and D-BP in a 
grid network topology under the \emph{2-hop} interference model. Finally, 
we compare throughput performance of Q-GMS and D-GMS in a size-6 ring 
network under the \emph{1-hop} interference model.

\begin{figure}[t]
\centering
% \subfigure[Network topology]{\epsfig{file=shortflow.eps,width=0.7\linewidth}\label{fig:lastpkt_a}}\\
% \subfigure[HOL delay of short flow $(2 \rightarrow 3)$ when $\lambda = 2$]{\epsfig{file=lastpkt.eps,width=0.8\linewidth}\label{fig:lastpkt_b}}
\subfigure[``H"-type network topology]{\epsfig{file=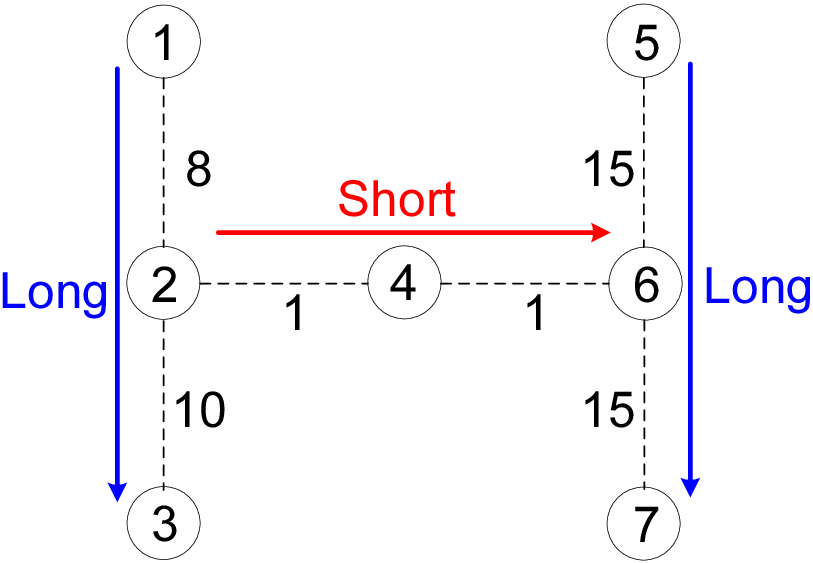,width=0.4\linewidth}\label{fig:lastpkt_a}} ~~~~~~~~~
\subfigure[HOL delay of short flow $(2 \rightarrow 4 \rightarrow 6)$ when $\lambda = 3$]{\epsfig{file=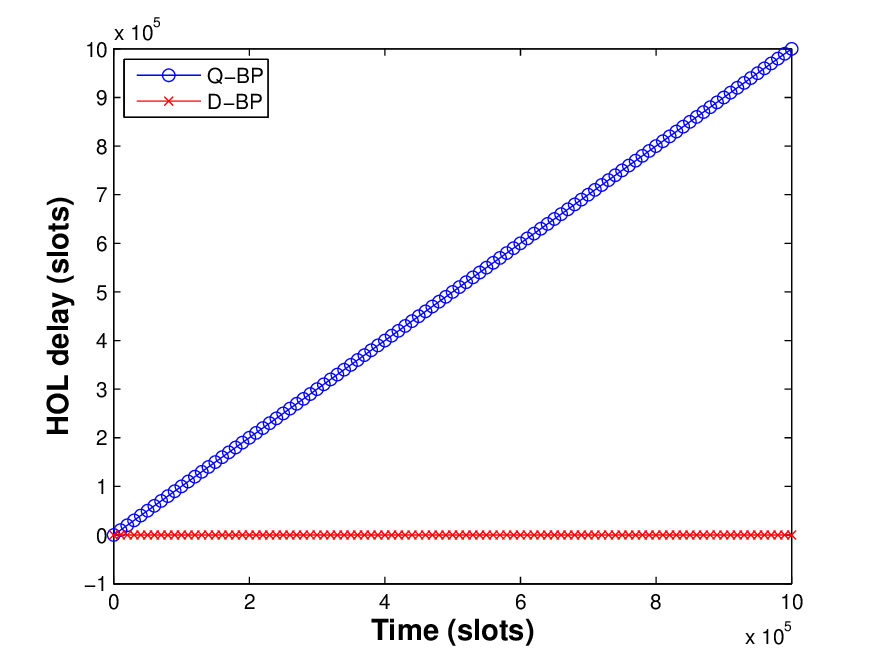,width=0.45\linewidth}\label{fig:lastpkt_b}}
\caption{Illustration of the last packet problem under Q-BP.}
\label{fig:lastpkt}
\end{figure}

We first show the last packet problem of Q-BP through simulations. 
We observe that several last packets of a short flow (that carry 
a finite amount of data) may get stuck, which could cause excessive 
\high{delays}. We consider a scenario consisting of 7 nodes and 6 links as 
shown in Fig.~\ref{fig:lastpkt_a}, where nodes are represented by 
circles and links are represented by dashed lines with their associated
link capacities\footnote{Unit of link capacity is packets per time slot.}. 
We assume a time-slotted system. We establish three flows: one short 
flow ($2 \rightarrow 4 \rightarrow 6$) and two long flows ($1 
\rightarrow 2 \rightarrow 3$) and ($5 \rightarrow 6 \rightarrow 7$). 
The short flow arrives in the network with 10 packets at time 0. The 
long flows have an infinite amount of data and keep injecting packets 
at the source nodes following Poisson distribution with mean rate 
$\lambda$ at each time slot. Numerical calculation shows that the 
feasible rate under the 2-hop interference should satisfy that $\lambda 
\le 4.44$. We conduct our simulation for $10^6$ time slots, and plot 
time traces of HOL delay of the short flow when $\lambda = 3$. 
Fig.~\ref{fig:lastpkt_b} illustrates the results that the delay increases 
linearly with time under Q-BP, which implies that several last packets 
of the short flow are excessively delayed. On the other hand, D-BP 
succeeds in serving the short flow and keeps the delay close to 0. This 
also implies that certain flows whose queue lengths do not increase 
due to lack of future arrivals (or whose inter-arrival times between 
groups of packets are very large) may experience a large delay under 
Q-BP, which will be confirmed in the following simulations. 

\begin{figure}[t]
\centering
\subfigure[Grid network topology]{\epsfig{file=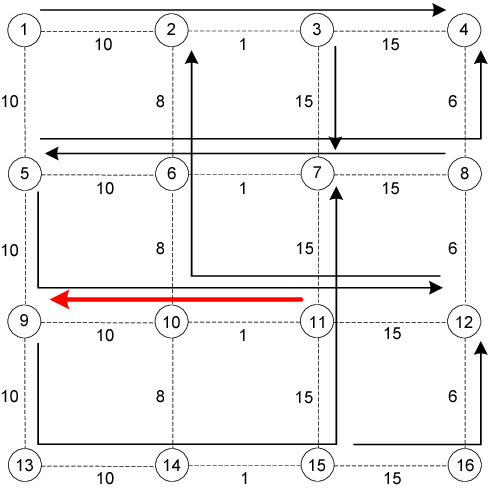,width=0.35\linewidth}\label{fig:multihop_a}} ~~~~~~~~~
\subfigure[Average queue length]{\epsfig{file=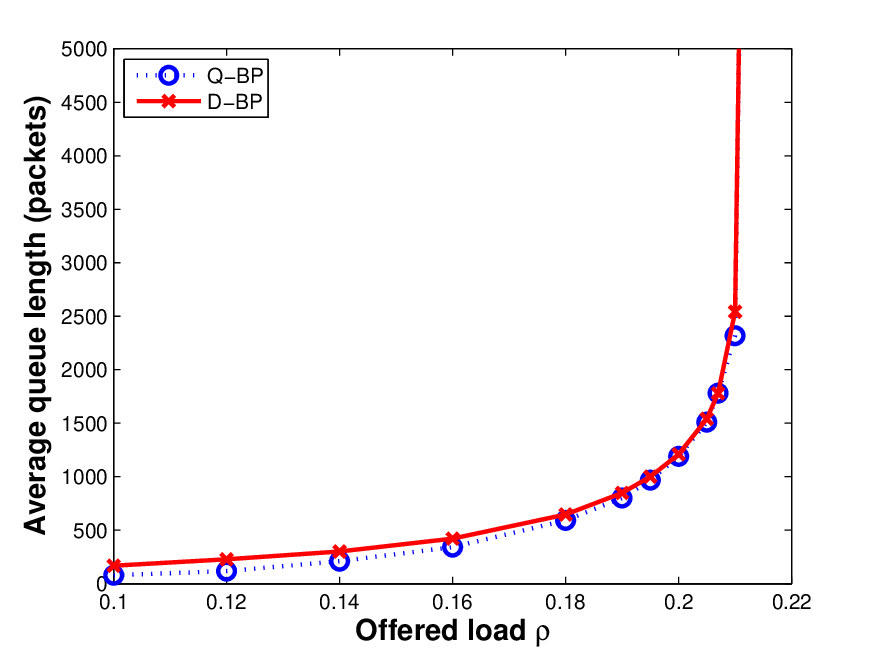,width=0.5\linewidth}\label{fig:multihop_b}}
% \subfigure[Network topology]{\epsfig{file=singlehopflow.eps,width=0.7\linewidth}\label{fig:single_flow_a}}\\
% \subfigure[Average delay]{\epsfig{file=singledelay_gms.eps,width=0.8\linewidth}\label{fig:single_flow_b}}
\caption{Performance of scheduling algorithms for multihop traffic following Poisson distribution.}
\label{fig:multihop}
\end{figure}

Next, we evaluate the throughput performance of different schedulers in a grid 
network that consists of 16 nodes and 24 links as shown in Fig.~\ref{fig:multihop_a}, 
where nodes and links are represented by circles and dashed lines, respectively, 
with link capacity. We establish 9 multihop flows that are represented by arrows. 
Let $\lambda_1=0.1$ and $\lambda_2=1$. At each time slot, there is a file arrival 
with probability $p=0.01$ for flow ($11 \rightarrow 10 \rightarrow 9$) (represented 
by the red thick arrow in Fig.~\ref{fig:multihop_a}), and the file size follows 
Poisson distribution with mean rate\footnote{Note that given the network topology, 
it is hard to find the exact boundary of the optimal throughput region of 
scheduling policies in a closed form. Hence, we probe the boundary by scaling the 
amount of traffic. After we choose $\vlambda$, which determines the direction of 
traffic load vector, we run our simulations with traffic load $\rho \vlambda$ changing 
$\rho$, which scales the traffic loads.} $\rho \lambda_1/p$. Note that flow 
($11 \rightarrow 10 \rightarrow 9$) has bursty arrivals with a small mean rate (we 
simply call it the bursty flow in the following part). All the other 8 flows have 
packet arrivals following Poisson distribution with mean rate $\rho \lambda_2$ at 
each time slot. Although these flows share the same stochastic property with an 
identical mean arrival rate $\rho \lambda_2$, uniform patterns of traffic are avoided 
by carefully setting the link capacities and placing the flows with different number 
of hops in an asymmetric manner. 

We evaluate the scheduling performance by measuring average total queue lengths in 
the network over time. Fig.~\ref{fig:multihop_b} illustrates average queue lengths 
under different offered loads to examine the performance limits of scheduling schemes. 
Each result represents an average of 10 simulation runs with independent stochastic 
arrivals, where each run lasts for $10^6$ time slots. Since the optimal throughput 
region is defined as the set of arrival rates under which the queue lengths remain 
finite, we can consider the traffic load, under which the queue length increases 
rapidly, as the boundary of the optimal throughput region. Fig.~\ref{fig:multihop_b} 
shows that D-BP achieves the same throughput region as Q-BP, thus supporting the 
theoretical results on throughput performance. 

\begin{figure}[t]
\centering
\epsfig{file=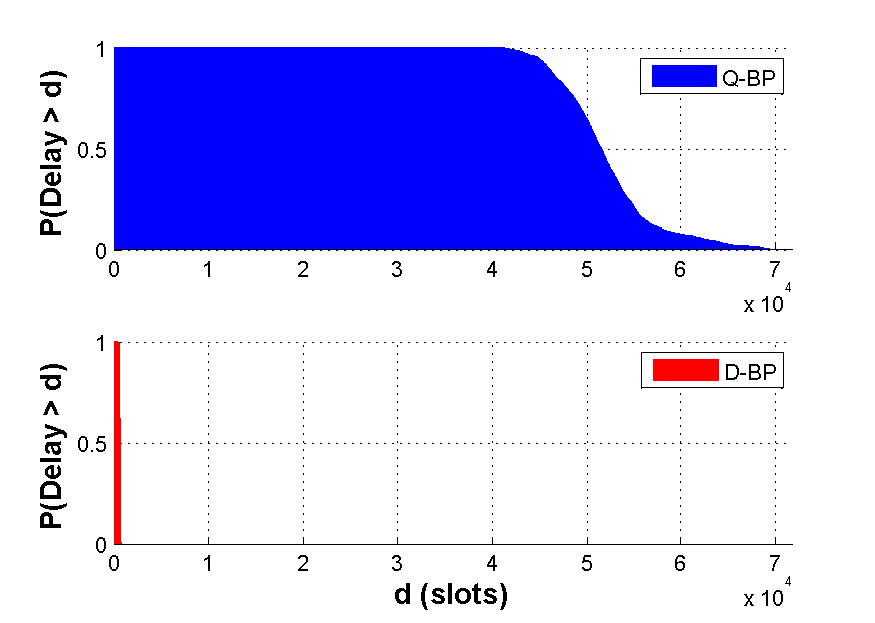,width=0.5\linewidth}
\caption{Delay distribution of the bursty flow under $\rho=0.2$.}
\label{fig:tail}
\end{figure}
\begin{figure}[t]
\centering
\epsfig{file=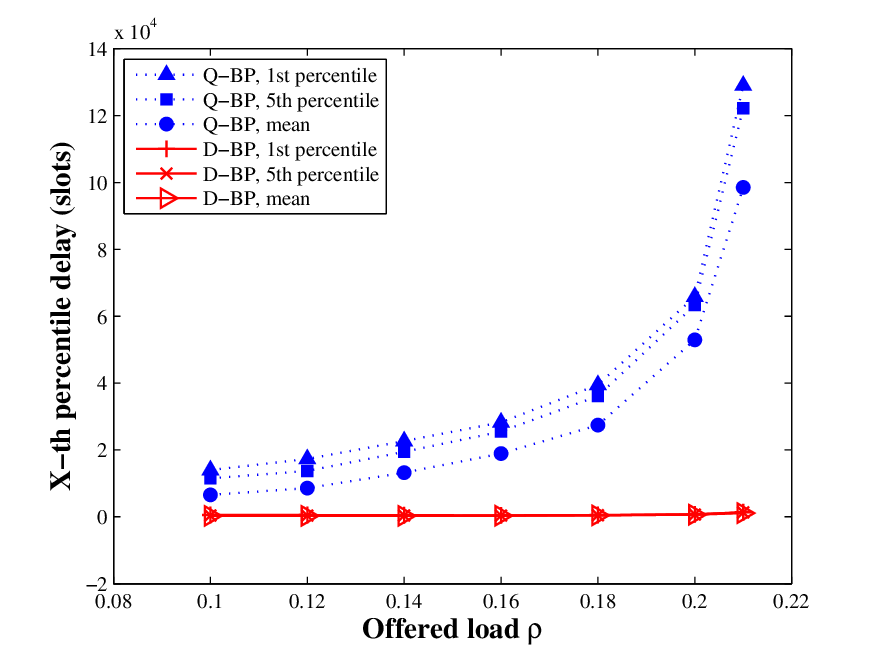,width=0.5\linewidth}
\caption{Mean delay, \high{the 1st and 5th percentile} delay of the bursty flow 
($11 \rightarrow 10 \rightarrow 9$) over offered loads.}
\label{fig:percentile}
\end{figure}

Although Q-BP and D-BP perform similarly in terms of the average queue length (or average 
delay due to Little's Law) over the network, the tail of the delay distribution of Q-BP 
could be substantially longer because certain flows are starved. This could cause enormous 
unfairness between flows, resulting in very poor QoS for certain flows. 

Note that although a bursty flow is a long flow that has an infinite amount of data, 
the arrivals occur in a dispersed manner (i.e., the inter-arrival times between groups 
of packets are very large) and we can view this bursty flow as consisting of many short 
flows. Thus, we expect that the bursty flow may experience a very large delay under Q-BP.
This is because the bursty flow lacks subsequent packet arrivals over long periods of 
time, which does not allow the queue-lengths to grow, and thus contributes to the long 
tail of the delay distribution.  However, this phenomenon may not manifest itself in 
terms of a higher average delay for Q-BP, as can be observed in Fig.~\ref{fig:multihop_b}, 
because the amount of data corresponding to the bursty flow in the simulation is small 
compared to the other flows. On the other hand, D-BP can achieve better fairness by 
scheduling the links based on delays and not starving bursty or variable flows. We 
confirm this in the following observations.

Fig.~\ref{fig:tail} illustrates the effectiveness of using D-BP over Q-BP in terms of how 
each scheme affects the delay distribution of bursty flows. We set $\rho=0.2$. The results
show that the tail of the delay distribution under D-BP vanishes much faster than Q-BP. 
Further, we plot the mean delay, \high{the 1st and 5th percentile} delay\footnote{Suppose 
there are $N$ packets sorted by their delays from the largest to the smallest, \high{the 
$X$-th percentile} delay is defined as the delay of the $\lfloor \frac{NX} {100} \rfloor$-th 
packet. If $\frac{NX} {100} \le 1$, it means the maximum delay. For example, if the delays 
are $[3,2,1,1,1]$, the 40th percentile delay is 2.} of the bursty flow over offered loads in 
Fig.~\ref{fig:percentile}. All these delays under D-BP are substantially less than under Q-BP, 
which implies that D-BP successfully eliminates the excessive packet delays. This confirms 
that, Q-BP causes a substantially long tail for the delay distribution of the network due 
to the starvation of the bursty flow, while D-BP overcomes this and achieves better fairness 
among the flows by scheduling the links based on delays.

\begin{figure}[t]
\centering
\subfigure[A size-6 ring network topology]{\epsfig{file=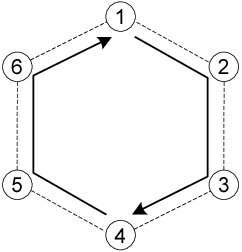,width=0.35\linewidth}\label{fig:c6_a}} ~~~~~~~~~
\subfigure[Average queue length]{\epsfig{file=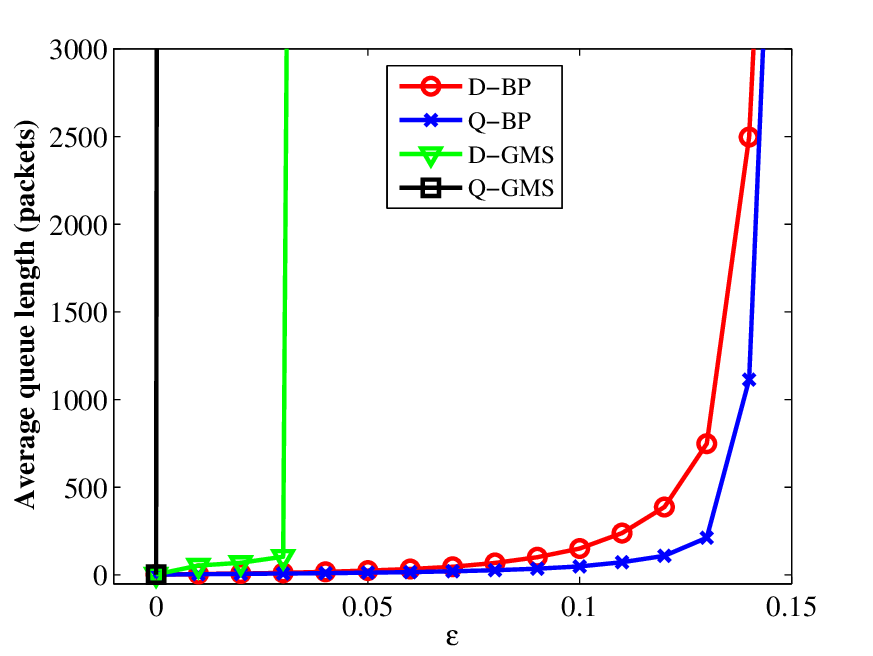,width=0.5\linewidth}\label{fig:c6_b}}
\caption{Performance comparison of Q-BP, D-BP, Q-GMS and D-GMS for multihop traffic under the 1-hop 
interference model.}
\label{fig:c6}
\end{figure}

Finally, we consider a size-6 ring network topology under the \emph{1-hop} interference 
model as shown in Fig.~\ref{fig:c6_a}, where links have unit link capacity. We simulate 
two flows: flow ($1 \rightarrow 2 \rightarrow 3 \rightarrow 4$) and flow ($4 \rightarrow 
5 \rightarrow 6 \rightarrow 1$). It is known \cite{joo09} that Q-GMS is not throughput-optimal
in this network, as the \emph{local pooling} condition is not satisfied (and thus the 
\emph{multihop local pooling} is not satisfied from Lemma 7 of \cite{zussman08}). On the 
other hand, although D-GMS is at least as efficient as Q-GMS, it is not known whether 
D-GMS can achieve larger throughput in certain scenarios, e.g., in the network in 
Fig.~\ref{fig:c6_a}. 

To see these, we construct a traffic pattern using the idea in \cite{joo09b}. 
We consider packet arrivals in a frame of 12 time slots. Two flows have the
same arrival pattern in each frame. We assume two arrival patterns for each
frame. Starting with empty queues at time slot 0, in each frame, the number 
of exogenous packet arrivals at the source of each flow (i.e., nodes 1 and 4) 
follows pattern $P_1=\{1,0,5,0,\,1,0,5,0,\,1,0,5,0\}$ with probability 
$\epsilon$, and pattern $P_2=\{1,0,0,\,1,0,0,\,1,0,0,\,1,0,0\}$ with 
probability $(1-\epsilon)$, where $0\le\epsilon\le1$. 
The average arrival rate vector is then $\vlambda = ( \frac {18} {12} \epsilon 
+ \frac {4} {12} (1-\epsilon))\textbf{e} = ( \frac {1} {3} + \frac {7} {6} \epsilon)
\textbf{e}$, where \textbf{e} is a dimension-2 vector with all components equal 
to 1. It is easy to check that $\vlambda$ lies strictly inside the optimal throughput 
region when $0 \le \epsilon < \frac {1} {7}$, while Q-GMS cannot stabilize the 
network under such a traffic pattern for all $\epsilon > 0$. Because under Q-GMS, 
when pattern $P_2$ occurs in a frame, all the packets arriving in this frame can 
be completely served and leave the network by the end of this frame, while pattern 
$P_1$ occurs, none of the packets arriving in this frame leaves the network \high{by} 
the end of this frame. We evaluate the performance 
of different scheduling policies under the above traffic pattern. For each policy 
under a fixed $\epsilon$, we take the average over 10 independent experiments, 
with each run being $10^7$ time slots. In Fig.~\ref{fig:c6_b}, we can see that 
Q-BP and D-BP have finite average queue length for $0 \le \epsilon < \frac 1 7 = 0.143$ 
and thus achieve the maximum throughput. On the other hand, the average queue 
length increases linearly with $\epsilon$ under Q-GMS and D-GMS starting from 
$\epsilon = 0$ and $\epsilon=0.04$, respectively. This implies that neither Q-GMS 
nor D-GMS is throughput-optimal in this setting, while D-GMS achieves larger throughput 
($\epsilon < 0.04$). To fully characterize the performance limits of D-GMS is 
an interesting yet challenging problem. 

%$0.1\%$ and $0.05\%$ largest delay of the network over offered loads
%in Fig.~\ref{fig:all}. The results show that Q-BP has much higher top $0.1\%$ and 
%$0.05\%$ largest delays, and this implies that the tail of the delay distribution of the 
%network is substantially longer under Q-BP. 
%%%%%%%%%%%%%%%%%%%%%%%%%%%%%%%%%%%%%%%%%%%
\section{Conclusion} \label{sec:conclusion}
%%%%%%%%%%%%%%%%%%%%%%%%%%%%%%%%%%%%%%%%%%%
In this paper, we developed a throughput-optimal delay-based back-pressure scheduling
scheme for multihop wireless networks \high{with fixed routes}. We introduced a new 
delay metric suitable for multihop traffic and established a linear relation between 
queue lengths and delays in the fluid limits, which plays a key role in the performance 
analysis and proof of throughput-optimality. Delay-based schemes provide a simple way 
around the well-known last packet problem that plagues queue-based schedulers, and thus 
avoid flow starvation. As a result, the excessively long delays that could be experienced 
by certain flows under queue-based scheduling schemes 
are eliminated without any loss of throughput. 
\high{
Nonetheless, in this paper, we have only considered the scheduling problem with 
fixed routes, albeit with multihop flows. The question of whether delay-based 
schemes under dynamic routing can achieve throughput-optimality is still very much
open.}

\appendices

\section{Summary of notations} \label{app:notation}
\begin{table}[t]
%\begin{color}{blue}
\begin{tabularx}{\linewidth}{lX}
\hline
Symbol & Definition \\
\hline
$\Vertex$ & set of nodes \\
$\Edge$ & set of links \\
$\Flow$ & set of flows \\
$\Pair$ & set of link-flow-pairs \\
$\Matching_\Pair$ & set of feasible schedules \\
$Co(\Matching_\Pair)$ & convex hull of $\Matching_\Pair$ \\
$\Lambda^{*}$ & optimal throughput region \\ 
$\Hop(s)$ & \# of hops on the route of flow $s$ \\
$\Hop^{\max}$ & $\max_{s \in \Flow} \Hop(s)$ \\
$A_s(t)$ & \# of packet arrivals for flow $s$ at time slot $t$ \\
$\lambda_s$ & mean arrival rate for flow $s$ \\
$F_s(t)$ & cumulative \# of packet arrivals for flow $s$ up to time slot $t$ \\
$Q_\sk(t)$ & queue length of $Q_\sk$ at time slot $t$ \\
$\Pi_\sk(t)$ & service at $Q_\sk$  at time $t$ \\
$\Psi_\sk(t)$ & \# of packet departures at $Q_\sk$ at time slot $t$ \\
$P_\sk(t)$ & $\sum_{i=1}^k Q_{s,i}(t)$ \\
$\hat{F}_\sk(t)$ & cumulative \# of packets served at $Q_\sk$ up to time slot $t$ \\
$Z_{\sk,i}(t)$ & sojourn time (in the network) of the $i$-th packet of $Q_\sk$ at time slot $t$ \\
$W_\sk(t)$ & sojourn time (in the network) of the HOL packet of $Q_\sk$ at time slot $t$, i.e.,  $Z_{\sk,1}(t)$\\
$U_\sk(t)$ & time when the HOL packet of $Q_\sk$ arrives in the network, i.e., $t-W_\sk(t)$ \\
$\hat{W}_\sk(t)$ & $W_\sk(t) - W_{s,k-1}(t)$ \\
$\Delta \hat{W}_\sk(t)$ & $\hat{W}_\sk(t) - \hat{W}_{s,k+1}(t)$ \\
$\Delta Q_\sk(t)$ & $Q_\sk(t) - Q_{s,k+1}(t)$ \\
$B_\sk(t)$ & inter-arrival time (at the system or the source node) between the HOL packet of $Q_\sk$ and the packet that arrives immediately after it \\
\hline
\end{tabularx}
\caption{Summary of notations}
\label{tab:notation}
%\end{color}
\end{table}

%%%%%%%%%%%%%%%%%%%%%%%%%%%%%%%%%%%%%%%%%%%%%%%%%%%%%%%%%%%%
\section{Proof of Lemma~\ref{lem:init}} \label{app:lem:init}
%%%%%%%%%%%%%%%%%%%%%%%%%%%%%%%%%%%%%%%%%%%%%%%%%%%%%%%%%%%%
%\subsection{Proof of Lemma~\ref{lem:init}} \label{app:lem:init}

\begin{proof}
We show that there exists a finite time $T>0$ such that the fluid limits 
satisfy $\hat{f}_{\sk}(T) > f_s(0)$ for all link-flow-pairs $(\sk) \in \Pair$.
We prove this by induction. We show that there exists a finite time $T$ with 
at least one link-flow-pair that satisfies the condition, and for a given set of 
link-flow-pairs satisfying the condition, at least one additional link-flow-pair 
will satisfy the condition by increasing $T$. 

We first fix an arbitrary $\epsilon_1 > 0$ and define a constant 
$K_1 \triangleq \max_s \Hop(s) + \left(\sum_s \lambda_s \Hop(s) \right) \epsilon_1$. 
In the fluid limit model, we will have
\[
f_s(\epsilon_1) = f_s(0) + \lambda_s \epsilon_1 > f_s(0),~\text{for all}~ s \in \Flow.
\]
Since queue lengths are no greater than the injected amount of data, we have that 
$p_{\sk}(\epsilon_1) \le f_s(\epsilon_1)$ for all $(\sk) \in \Pair$, and thus, 
\begin{equation}
\begin{split}
\label{eq:k1}
\textstyle \sum_{(\sk) \in \Pair} p_\sk(\epsilon_1) &\le \textstyle \sum_{(\sk) \in \Pair} f_s(\epsilon_1) \\
&\le \textstyle \sum_s \Hop(s) \left( f_s(0) + \lambda_s \epsilon_1 \right) \\
&\le K_1,
\end{split}
\end{equation}
where the last inequality is from Eq. (\ref{eq:initial}): $\sum_s f_s(0) \le 1$ and 
the definition of $K_1$. Now we show by induction that there exists a finite time 
$T$ such that
\[
\hat{f}_\sk(T) > f_s(0),~\text{for all link-flow-pairs}~(\sk).
\]

\noindent {\bf Base Case:}  
There exists $T_1>0$ such that for at least one link-flow-pair $(\sk)$,
\begin{equation}
\label{eq:base}
\hat{f}_\sk(T_1) \ge f_s(\epsilon_1).
\end{equation}

\high{
Let $T_1 \triangleq \epsilon_1 + K_1$. Suppose that (\ref{eq:base}) does not hold,
i.e., there exists at least one packet that arrives before time slot $\lfloor \xnj 
\epsilon_1 \rfloor + 1$ and is not served by the end of time slot $\lfloor \xnj T_1 
\rfloor$. Hence, at each time slot between $[\lfloor \xnj \epsilon_1 \rfloor + 1, 
\lfloor \xnj T_1 \rfloor ]$, there exists at least one schedule that has positive
summed weight. Therefore, the schedule determined by D-BP must serve at least one 
packet in the original system, otherwise the summed weight of the schedule (that 
does not serve any packet) is zero, which is not the maximum over all the feasible 
schedules. Hence, we must have
\[
\begin{split}
\textstyle \sum_{(\sk) \in \Pair} \left( \hat{F}^{(\xnj)}_\sk( \xnj T_1 ) - 
\hat{F}^{(\xnj)}_\sk( \xnj \epsilon_1 ) \right) \\
 \textstyle \ge \lfloor \xnj T_1 \rfloor - \lfloor \xnj \epsilon_1 \rfloor,
\end{split}
\]
Dividing both sides of the above inequality by $\xnj$ and letting $\xnj \rightarrow 
\infty$, we obtain
\[
\textstyle \sum_{(\sk) \in \Pair} \left( \hat{f}_\sk(T_1) - 
\hat{f}_\sk(\epsilon_1) \right) \ge K_1.
\]
}
Then, from (\ref{eq:k1}), we have
\[
\begin{split}
%\begin{eqnarray}
\textstyle \sum_{(\sk) \in \Pair} \hat{f}_\sk(T_1) &\ge 
\textstyle \sum_{(\sk) \in \Pair} \hat{f}_\sk(\epsilon_1) +
\sum_{(\sk) \in \Pair} p_\sk(\epsilon_1) \\
&= \textstyle \sum_{(\sk) \in \Pair} f_s(\epsilon_1). 
%\end{eqnarray}
\end{split}
\]
Therefore, $\hat{f}_\sk(T_1) \ge f_s(\epsilon_1)$ for at least one link-flow-pair
$(\sk)$.

\noindent {\bf Inductive Step:} 
Suppose that there exist $T_l$ and a subset $S_l \subseteq \Pair$ such that for all 
$(\sk) \in S_l$, we have 
\begin{equation}
\label{eq:fixed}
\hat{f}_\sk(T_l) \ge f_s(\epsilon_1).
\end{equation}
Then there \high{exists} $T_{l+1} \ge T_l$, where $1 \le l < \sum_s \Hop(s)$, and a 
link-flow-pair $(\tilde{s}, \tilde{k}) \in \Pair \backslash S_l$ such that
\begin{equation}
\label{eq:induction}
\hat{f}_{\tilde{s}, \tilde{k}}(T_{l+1}) \ge f_{\tilde{s}}(\epsilon_1).
\end{equation}
Further we define $S_{l+1} = S_l \cup \{(\tilde{s}, \tilde{k})\}$. 

We prove the inductive step for $l=1$. The generalization for $l>1$ is straightforward.  
Hence, we show that for given $S_1$ and $T_1$, there exists a finite $T_2 > T_1$ such 
that (\ref{eq:induction}) with $T_2$ holds for at least two different link-flow-pairs.

Let $(\hsk)$ denote the link-flow-pair that satisfies (\ref{eq:fixed}) with $T_1$. 
Then, we have\footnote{Note that if $(s,k) \in S_l$, we must have $(s,k-1) \in S_l$. 
Hence, for $l=1$, we must have the first hop of a flow, i.e., $S_1 = {(\hat{s},1)}$ 
for some $\hat{s}$.} $S_1 = \{ (\hat{s},1) \}$ and can specify the set $S_1^*$ of 
link-flow-pairs $(s,k) \in \Pair \backslash S_1$ that is closest to the source of each 
flow from (\ref{eq:s1ss}). We illustrate the case that $\Hop(\hat{s}) > 1$, and the other 
case that $\Hop(\hat{s}) = 1$ can be easily shown following the same line of analysis. 
Now, we have 
\[
\hat{f}_{\hat{s},1}(t) \ge f_{\hat{s}}(\epsilon_1), ~\text{for all}~t \ge T_1.
\]
For all the other link-flow-pairs, we observe that
\begin{equation}
\label{eq:k1k1}
\textstyle \sum_{(\rj) \in \Pair \backslash S_1} \left( f_r(\epsilon_1) - 
\hat{f}_{\rj}(T_1) \right) \le K_1.
\end{equation}
% and 
% \begin{equation}
% q_{\hat{s},1}(T_1) \le K_0,
% \end{equation}
% where $K_0 \triangleq 1 + \lambda_{\hat{s}} T_1$.

Suppose that for all $t \ge T_1$, we have
\begin{equation}
\label{eq:assum}
\hat{f}_{\rj}(t) < f_r(\epsilon_1), ~\text{for all}~ (\rj) \in \Pair \backslash S_1.
\end{equation}
In the following part, we provide a choice of $T_2 > T_1$ such that assumption 
(\ref{eq:assum}) leads to a contradiction, which completes the inductive step, 
and then the lemma follows by induction.

We view each sample path $X^{(\xnj)}(t)$ after time slot $\lceil \xnj T_1 \rceil$ 
as a generalized system with link-flow-pairs in $S_1 = \{(\hat{s},1)\}$. We say 
that a time slot is \emph{unavailable} to $S_1$ when a packet from a link-flow-pair 
$(\rj)\in \Pair \backslash S_1$ is transmitted during the time slot. Let $h_{S_1}(t)$ 
denote the (scaled) amount of time unavailable to $S_1$ during the period of $(T_1, t]$ 
in the scaled system, for all $t > T_1$. For the scaled generalized system $S_1$, 
we obtain from (\ref{eq:k1k1}) and (\ref{eq:assum}) that 
\begin{equation}
\label{eq:unavailable}
\textstyle h_{S_1}(t) \le \sum_{(\rj) \in \Pair \backslash S_1} \left( \hat{f}_{\rj}(t) - 
\hat{f}_{\rj}(T_1) \right) \le K_1,
\end{equation}
for all $t > T_1$. Since the time unavailable to $S_1$ is bounded, 
as time $t$ increases, only link-flow-pairs in $S_1$ will be scheduled, 
which implies that the weight of link-flow-pairs of $\Pair \backslash 
S_1$ becomes negligible. This allows us to focus on $S_1$. Owing to 
Lemma~\ref{lem:qw} and the definition of $S_1$, the linear relation 
between queue lengths and delays holds for the link-flow-pair in $S_1$. 
Then, it can be easily shown following the same line of analysis of 
Proposition~\ref{pro:dbp} that link-flow-pairs in $S_1$ are stable 
under D-BP\footnote{Note that since Lemmas~\ref{lem:qw} and~\ref{lem:init} 
hold for the generalized system $S_1$, Proposition~\ref{pro:dbp} can 
be applied to $S_1$.}. Hence, for all $(\sk) \in S_1$, we have 
\begin{equation}
\label{eq:qb}
q_{\sk}(t) \le C_1, ~\text{for all}~ t\ge T_1, 
\end{equation}
and thus
\begin{equation}
\label{eq:wb}
\hat{w}_{\sk}(t) \le \frac {C_1} {\lambda_{s}},
~\text{for all}~ t\ge T_1, 
\end{equation}
for some constant $C_1$, which depends on $T_1$ and $K_1$ and does not depend on 
time $t$. 

Recall that $S_1^*$ denotes the set of link-flow-pairs that is closest to the source 
of each flow out of $S_1$ defined in (48). We choose $t$ large enough such that for 
all $(\sk) \in S_1$ and $(r^*,j^*) \in S_1^*$, 
\begin{equation}
\label{eq:c1}
\frac {C_1} {\lambda_{s}} - \left( t - \epsilon_1 
-  \frac {C_1} {\lambda_{s}} \right)< \left( t - \epsilon_1
- \frac {C_1} {\lambda_{r^*}} \right) - \epsilon_1.
\end{equation}
From (\ref{eq:assum}), there are packets that arrive at the source node by time $\epsilon_1$ 
and have not been served at $j$-th hop by time $t$ for all $(\rj) \in \Pair \backslash S_1$, 
we obtain that
\begin{equation}
\label{eq:c4}
t - \epsilon_1 \le w_{\rj}(t) \le t,
~\text{for all}~ (\rj) \in \Pair \backslash S_1.
\end{equation}
Since $(r^*,j^*), (r^*,j^*+1) \in P\backslash S_1$ for $(r^*,j^*) \in S_1^*$, we have
\begin{equation}
\label{eq:c3}
\hat{w}_{r^*,j^*+1}(t) = w_{r^*,j^*+1}(t) - w_{(r^*,j^*)}(t) \le \epsilon_1, 
\end{equation}
for all $(r^*,j^*) \in S^*_1$.
From (\ref{eq:wb}), (\ref{eq:c4}), and that $(r^*,j^*-1) \in S_1$, 
we have
\begin{equation}
\label{eq:c2}
\hat{w}_{r^*,j^*}(t) \ge t - \epsilon_1 - \frac {C_1} {\lambda_{r^*}},
\end{equation}
for all $(r^*,j^*) \in S^{*}_1$.
Then, we have
\[
\begin{split}
\Delta \hat{w}_{s,k}(t) &= \hat{w}_{s,k}(t) - \hat{w}_{s,k+1}(t) \\
&\stackrel{(a)}\le C_1 / \lambda_s - (t - \epsilon_1 - C_1 / \lambda_s) \\
&\stackrel{(b)}< (t - \epsilon_1 - C_1 / \lambda_{r^*}) - \epsilon_1 \\
&\stackrel{(c)}\le \hat{w}_{r^*,j^*}(t) - \hat{w}_{r^*,j^*+1}(t) \\
&= \Delta \hat{w}_{r^*,j^*} (t)
\end{split}
\]
for all $(s,k) \in S_1$ and $(r^*,j^*) \in S_1^*$, where (a) is from 
(\ref{eq:wb}) and (\ref{eq:c2}), (b) is from (\ref{eq:c1}), and (c) is
from (\ref{eq:c2}) and (\ref{eq:c3}). Hence, for large t, we have that 

\begin{equation}
\label{eq:skrj}
\Delta \hat{w}_{s,k}(t) < \min_{(r^*,j^*) \in S^*_1} 
\{ \Delta  \hat{w}_{r^*,j^*}(t) \}.
\end{equation}
Also, from (\ref{eq:c4}), we have that 
\begin{equation}
\label{eq:hatrj}
\Delta  \hat{w}_{\hat{r}, \hat{j}}(t) \le \epsilon_1,
\end{equation}
for all $(\hat{r}, \hat{j}) \in \Pair \backslash ( S_1 \cup S^*_1 )$.
Since (\ref{eq:hatrj}) holds for an arbitrarily small $\epsilon_1$ and 
from (\ref{eq:skrj}), D-BP favors link-flow-pairs of $S_1^*$ for all 
large $t$. Note that $\Delta\hat{w}_{s,k}(t)$ is bounded for $(s,k) \in S_1$ 
from (\ref{eq:wb}), and $\Delta\hat{w}_{\hat{r},\hat{j}}(t)$ is bounded 
for $(\hat{r},\hat{j}) \in P\backslash (S_1 \cup S_1^*)$ from (\ref{eq:hatrj}), 
and $\Delta\hat{w}_{r*,j*}(t)$ increases linearly on the order of $t$ for 
$(r^*,j^*) \in S_1^*$ from (\ref{eq:c2}). Hence, there exists a large 
$T_2^{\prime}$ such that for all $t > T_2^{\prime}$, link-flow-pairs in 
$S_1^*$ will be scheduled at all the time slots between $[\lfloor \xnj 
T_2^{\prime} \rfloor + 1, \lfloor \xnj t \rfloor]$ under D-BP. Then, we 
can choose $T_2 > T_2^{\prime}$ and have that 
\[
h_{S_1}(T_2) \ge T_2 - T_2^{\prime} > K_1.
\]
However, this contradicts \high{with} (\ref{eq:unavailable}), which shows that,
the assumption (\ref{eq:assum}) is false, and there exists a large $T_2$ 
such that 
\begin{equation}
\label{eq:contradiction}
\hat{f}_{\tilde{s}, \tilde{k}}(T_2) \ge f_{\tilde{s}}(\epsilon_1), 
~\text{for at least one}~ (\tilde{s}, \tilde{k}) \in \Pair \backslash S_1.
\end{equation}

In fact, our choice of $T_2$ depends on the set $S_1$. However, since there 
are only a finite number of flows, we can always choose a large enough $T_2$ 
so that (\ref{eq:contradiction}) holds for some $(\tilde{s}, \tilde{k}) \in 
\Pair \backslash S_1$.
\end{proof}

%%%%%%%%%%%%%%%%%%%%%%%%%%%%%%%%%%%%%%%%%%%%%%%%%%%%%%%%
\section{Lemma~\ref{lem:pd}} \label{app:lem:pd}
%%%%%%%%%%%%%%%%%%%%%%%%%%%%%%%%%%%%%%%%%%%%%%%%%%%%%%%%
\begin{lemma}
\label{lem:pd}
Consider a system under the Q-BP policy. Then for $\vlambda$ strictly inside 
$\Lambda^{*}$, there exists a finite time $T > 0$ such that the fluid limits 
satisfy the following property for all $t \ge T$ with probability one:
\begin{equation}
\label{eq:lem:dqg}
q_\sk(t) \ge q_{s,k+1}(t), ~\text{i.e.},~\Delta q_\sk(t) \ge 0,
\end{equation}
for all $(\sk) \in \Pair$.
\end{lemma}

\begin{proof}
We let $\Pair_{\sk} \triangleq \{ (s,j)~|~1 \le j \le k\}$ denote the
set of link-flow-pairs among the first $k$ hops of flow $s$. Consider 
a flow $\hat{s} \in \Flow$. We want to show that there exists a finite 
time $T > 0$ such that for all time $t \ge T$, (\ref{eq:lem:dqg}) holds 
for every link-flow-pair $(\hat{s},k) \in \Pair_{\hat{s},\Hop(\hat{s})}$. 
We prove it by induction.

\noindent {\bf Base Case:}  
We first show that there exists a finite time $T_{\hat{s},1}>0$ such that 
(\ref{eq:lem:dqg}) holds for $(\hat{s},1)$ and for any $t \ge T_{\hat{s},1}$.
Suppose that (\ref{eq:lem:dqg}) does not hold for $(\hat{s},1)$ and for all
$t \ge 0$. Then, Q-BP does not schedule link-flow-pair $(\hat{s},1)$ due to 
the operation of Q-BP that it does not schedule any link-flow-pair $(\sk)$ 
with $\Delta Q_\sk(t) < 0$. On the other hand, due to the exogenous arrivals
at the source node of flow $\hat{s}$, $Q_{\hat{s},1}(t)$ must increase with 
time. Specifically, let $T_{\hat{s},1} \triangleq 1/\lambda_{\tilde{s}}$, 
then we have $q_{\tilde{s},1}(T_{\hat{s},1}) = q_{\tilde{s},1}(0) + 
\lambda_{\tilde{s}} T_{\hat{s},1} \ge 1$. Since Q-BP does not schedule 
link-flow-pair $(\hat{s},1)$, then it satisfies that $q_{\tilde{s},2}(T_{\hat{s},1}) 
\le q_{\tilde{s},2}(0) \le 1$ from (\ref{eq:initial}). Hence, $\Delta 
q_{\tilde{s},1}(T_{\hat{s},1}) = q_{\tilde{s},1}(T_{\hat{s},1}) - 
q_{\tilde{s},2}(T_{\hat{s},1}) \ge 0$, i.e., (\ref{eq:lem:dqg}) holds for 
link-flow-pair $(\tilde{s},1)$ at time $T_{\hat{s},1}$. We next show that 
(\ref{eq:lem:dqg}) also holds for all $t \ge T_{\hat{s},1}$ for link-flow-pair 
$(\tilde{s},1)$ under Q-BP. Suppose that $t^* > T_{\hat{s},1}$ is the first
time after $T_{\hat{s},1}$ such that $\Delta q_{\tilde{s},1}(t^*) < 0$ occurs.
Consider a positive sequence $\{\xnj\}$ for which the convergence to the 
fluid limits holds. Then $Q_{\tilde{s},1}$ is scheduled at some time slots in the 
interval of $[ \lfloor \xnj T_{\hat{s},1} \rfloor + 1, \lfloor \xnj t \rfloor]$
in the original system. Let $\tau^*$ be the first such time slot in the interval 
of $[ \lfloor \xnj T_{\hat{s},1} \rfloor + 1, \lfloor \xnj t \rfloor]$ when 
$Q_{\tilde{s},1}$ is scheduled in the original system. Hence, we have 
$Q_{\tilde{s},1}(\tau^*) \ge Q_{\tilde{s},2}(\tau^*)$, otherwise it is not scheduled. 
This further implies that $Q_{\tilde{s},1}(\tau) \ge Q_{\tilde{s},2}(\tau) - 2$ 
for any time slot $\tau \ge \tau^*$, following a similar argument for showing 
(\ref{eq:qdcon}). Therefore, we must have $\Delta q_{\tilde{s},1}(t^*) \ge 0$ 
from the convergence of (\ref{eq:fluid_q}), which leads to a contradiction. 
Thus, (\ref{eq:lem:dqg}) holds for any $t \ge T_{\hat{s},1}$ for link-flow-pair 
$(\tilde{s},1)$ under Q-BP.

\noindent {\bf Inductive Step:} 
Suppose that there exists a finite time $T_{\hat{s},k}>0$ such that (\ref{eq:lem:dqg}) 
holds for all $(s,k) \in \Pair_{\hat{s},k}$ and for all $t \ge T_{\hat{s},k}$,
where $1 \le k < \Hop(\hat{s})$, we want to show that there exists a finite 
time $T_{\hat{s},k+1} \ge T_{\hat{s},k}$ such that (\ref{eq:lem:dqg}) holds 
for all $(s,k) \in \Pair_{\hat{s},k+1}$ and for all $t \ge T_{\hat{s},k+1}$. 
Clearly, it is sufficient to show that (\ref{eq:lem:dqg}) holds for 
$(\hat{s},k+1)$ for all $t \ge T_{\hat{s},k+1}$.

Let $\Pair_0(\tau)$ denote the set of link-flow-pairs such that (\ref{eq:lem:dqg}) 
holds for all $t \ge \tau$. Clearly, we have $\Pair_{\hat{s},k} \subseteq 
\Pair_0(T_{\hat{s},k})$. Suppose $\Pair_0(T_{\hat{s},k}) = \Pair_0(t)$ for all 
$t \ge T_{\hat{s},k}$, i.e., the set of link-flow-pairs for which (\ref{eq:lem:dqg}) 
holds does not change after time $T_{\hat{s},k}$. Then, Q-BP will schedule only 
link-flow-pairs in set $\Pair_0(T_{\hat{s},k})$ for all time slot $t \ge \lfloor 
\xnj T_{\hat{s},k} \rfloor + 1$ in the original system. This implies that the 
fluid limit model of the subsystem that consists of link-flow-pairs in 
$\Pair_0(T_{\hat{s},k})$ must be stable for any $\vlambda$ strictly inside $\Lambda^*$, 
from throughput-optimality of Q-BP (See Proposition~\ref{pro:qbp}). Specifically, 
we can show that for any fixed $\zeta>0$, there exists a $T^{\prime} \ge T_{\hat{s},k}$ 
such that $\sum_{(\sk) \in \Pair_0(T_{\hat{s},k})} q_\sk(t) \le \zeta$ for all $t \ge 
T^{\prime}$. We now consider two cases:
\begin{enumerate}
\item there is no link-flow-pair in set $\Pair \backslash \Pair_0(T_{\hat{s},k})$ 
that becomes satisfying (\ref{eq:lem:dqg}) by time $\max \{T^{\prime}, (1+\zeta)/\lambda_{\hat{s}}\}$;
\item there exists at least one link-flow-pair in set $\Pair \backslash \Pair_0(T_{\hat{s},k})$
that becomes satisfying (\ref{eq:lem:dqg}) by time $\max \{T^{\prime}, (1+\zeta)/\lambda_{\hat{s}}\}$.
\end{enumerate}

\noindent In {\bf Case 1)}, choose $T_{\hat{s},k+1} \triangleq \max \{T^{\prime}, 
(1+\zeta)/\lambda_{\hat{s}}\}$. Then,
\[
\begin{split}
q_{\hat{s},k+1}(T_{\hat{s},k+1}) &\ge f_{\hat{s}}(T_{\hat{s},k+1}) - \sum_{(\sk) \in \Pair_{\sk}} q_\sk(T_{\hat{s},k+1}) \\
&\ge f_{\hat{s}}(T_{\hat{s},k+1}) - \sum_{(\sk) \in \Pair_0(T_{\hat{s},k})} q_\sk(T_{\hat{s},k+1}) \\
&\ge  T_{\hat{s},k+1} \lambda_{\hat{s}} - \zeta \\
&\ge 1.
\end{split}
\]
Since Q-BP does not schedule link-flow-pair $(\hat{s},k+1)$ by time $T_{\hat{s},k+1}$, 
then it satisfies $q_{\tilde{s},k+2}(T_{\hat{s},k+1}) \le q_{\tilde{s},k+2}(0) \le 1$ 
from (\ref{eq:initial}). Hence, $\Delta q_{\tilde{s},k+1}(T_{\hat{s},k+1}) = 
q_{\tilde{s},k+1}(T_{\hat{s},k+1}) - q_{\tilde{s},k+2}(T_{\hat{s},1}) \ge 0$, i.e., 
(\ref{eq:lem:dqg}) holds for link-flow-pair $(\tilde{s},k+1)$ at time $T_{\hat{s},k+1}$.
Similar as in the base case, we can show that (\ref{eq:lem:dqg}) holds for any 
$t \ge T_{\hat{s},k+1}$ for link-flow-pair $(\tilde{s},k+1)$ under Q-BP.

\noindent In {\bf Case 2)}, let $T_0$ be the first time after $T_{\hat{s},k}$ when there 
is a link-flow-pair $(\tilde{s},\tilde{k}) \in \Pair \backslash \Pair_0(T_{\hat{s},k})$
such that (\ref{eq:lem:dqg}) holds for $(\tilde{s},\tilde{k})$ at time $T_0$. Suppose 
$\Pair_0(T_0) = \Pair_0(t)$ for all $t \ge T_0$, i.e., the set of link-flow-pairs for 
which (\ref{eq:lem:dqg}) holds does not change after time $T_0$. Then similarly, we can 
show that there exists $T^{\prime \prime} \ge T_0$ such that $\sum_{(\sk) \in \Pair_0(T_0)} 
q_\sk(t) \le \zeta$ for any $t \ge T^{\prime \prime}$. Again, we consider two cases:
\begin{enumerate}[i)]
\item there is no link-flow-pair in the set of $\Pair \backslash \Pair_0(T_0)$ 
that becomes satisfying (\ref{eq:lem:dqg}) by time $\max \{T^{\prime \prime}, (1+\zeta)/\lambda_{\hat{s}}\}$;
\item there exists at least one link-flow-pair in the set of $\Pair \backslash \Pair_0(T_0)$
that becomes satisfying (\ref{eq:lem:dqg}) by time $\max \{T^{\prime \prime}, (1+\zeta)/\lambda_{\hat{s}}\}$.
\end{enumerate}
In {\bf Case 2-i)}, we choose $T_{\hat{s},k+1} \triangleq \max \{T^{\prime \prime}, 
(1+\zeta)/\lambda_{\hat{s}}\}$. Following a similar argument in Case 1), we show 
that (\ref{eq:lem:dqg}) holds for all $t \ge T_{\hat{s},k+1}$ for link-flow-pair 
$(\tilde{s},k+1)$ under Q-BP. Since there are finite number of link-flow-pairs in 
the system, in {\bf Case 2-ii)}, recursively applying the above argument , we show 
that there exists a finite time $T_{\hat{s},k+1}$ such that (\ref{eq:lem:dqg}) 
holds for all $t \ge T_{\hat{s},k+1}$ for link-flow-pair $(\tilde{s},k+1)$ under Q-BP.

Choose $T_{\hat{s}} \triangleq \max_{1 \le k \le \Hop(\hat{s})} T_{\hat{s},k}$, 
then (\ref{eq:lem:dqg}) holds for all link-flow-pairs $(\sk) \in \Pair_{\hat{s},
\Hop(\hat{s})}$, for all time $t \ge T_{\hat{s}}$. 

Note that the above argument applies to any $\hat{s} \in \Flow$. Choose 
$T \triangleq \max_{s \in \Flow} T_s > 0$. Therefore, (\ref{eq:lem:dqg}) 
holds for all link-flow-pairs of $\Pair$ for all time $t \ge T$.
\end{proof}

\bibliographystyle{IEEEtran}
\bibliography{jibo}

\begin{thebibliography}{10}
\providecommand{\url}[1]{#1}
\csname url@samestyle\endcsname
\providecommand{\newblock}{\relax}
\providecommand{\bibinfo}[2]{#2}
\providecommand{\BIBentrySTDinterwordspacing}{\spaceskip=0pt\relax}
\providecommand{\BIBentryALTinterwordstretchfactor}{4}
\providecommand{\BIBentryALTinterwordspacing}{\spaceskip=\fontdimen2\font plus
\BIBentryALTinterwordstretchfactor\fontdimen3\font minus
  \fontdimen4\font\relax}
\providecommand{\BIBforeignlanguage}[2]{{%
\expandafter\ifx\csname l@#1\endcsname\relax
\typeout{** WARNING: IEEEtran.bst: No hyphenation pattern has been}%
\typeout{** loaded for the language `#1'. Using the pattern for}%
\typeout{** the default language instead.}%
\else
\language=\csname l@#1\endcsname
\fi
#2}}
\providecommand{\BIBdecl}{\relax}
\BIBdecl

\bibitem{tassiulas92}
L.~Tassiulas and A.~Ephremides, ``{Stability properties of constrained queueing
  systems and scheduling policies for maximum throughput in multihop radio
  networks},'' \emph{IEEE Transactions on Automatic Control}, vol.~37, no.~12,
  pp. 1936--1948, 1992.

\bibitem{lin06c}
X.~Lin, N.~B. Shroff, and R.~Srikant, ``{A tutorial on cross-layer optimization
  in wireless networks},'' \emph{IEEE Journal on Selected Areas in
  Communications}, vol.~24, no.~8, pp. 1452--1463, Aug. 2006.

\bibitem{warrier09}
A.~Warrier, S.~Janakiraman, S.~Ha, and I.~Rhee, ``{DiffQ: Practical
  differential backlog congestion control for wireless networks},'' in
  \emph{The 28th IEEE International Conference on Computer Communications
  (INFOCOM)}, 2009.

\bibitem{sridharan09}
A.~Sridharan, S.~Moeller, and B.~Krishnamachari, ``{Implementing
  Backpressure-based Rate Control in Wireless Networks},'' in \emph{Information
  Theory and Applications Workshop}, 2009.

\bibitem{moeller10}
S.~Moeller, A.~Sridharan, B.~Krishnamachari, and O.~Gnawali, ``{Routing without
  routes: The backpressure collection protocol},'' in \emph{Proceedings of the
  9th ACM/IEEE International Conference on Information Processing in Sensor
  Networks (IPSN)}, 2010, pp. 279--290.

\bibitem{van09}
P.~van~de Ven, S.~Borst, and S.~Shneer, ``{Instability of MaxWeight Scheduling
  Algorithms},'' in \emph{The 28th IEEE International Conference on Computer
  Communications (INFOCOM)}, 2009, pp. 1701--1709.

\bibitem{liu10a}
S.~Liu, L.~Ying, and R.~Srikant, ``{Throughput-Optimal Opportunistic Scheduling
  in the Presence of Flow-Level Dynamics},'' in \emph{The 29th IEEE
  International Conference on Computer Communications (INFOCOM)}, 2010.

\bibitem{liu10b}
------, ``{Scheduling in multichannel wireless networks with flow-level
  dynamics},'' \emph{ACM SIGMETRICS Performance Evaluation Review}, vol.~38,
  no.~1, pp. 191--202, 2010.

\bibitem{mekkittikul96}
A.~Mekkittikul and N.~McKeown, ``{A starvation-free algorithm for achieving
  100\% throughput in an input-queued switch},'' in \emph{Proc. of the IEEE
  International Conference on Communication Networks (ICCCN)}, 1996.

\bibitem{andrews01}
M.~Andrews, K.~Kumaran, K.~Ramanan, A.~Stolyar, P.~Whiting, and R.~Vijayakumar,
  ``{Providing quality of service over a shared wireless link},'' \emph{IEEE
  Communications magazine}, vol.~39, no.~2, pp. 150--154, 2001.

\bibitem{shakkottai02}
S.~Shakkottai and A.~Stolyar, ``{Scheduling for multiple flows sharing a
  time-varying channel: The exponential rule},'' \emph{Translations of the
  American Mathematical Society-Series 2}, vol. 207, pp. 185--202, 2002.

\bibitem{andrews04}
M.~Andrews, K.~Kumaran, K.~Ramanan, A.~Stolyar, R.~Vijayakumar, and P.~Whiting,
  \emph{{Scheduling in a queuing system with asynchronously varying service
  rates}}.\hskip 1em plus 0.5em minus 0.4em\relax Cambridge Univ Press, 2004,
  vol.~18.

\bibitem{sadiq09}
B.~Sadiq and G.~de~Veciana, ``{Throughput optimality of delay-driven MaxWeight
  scheduler for a wireless system with flow dynamics},'' in \emph{Proceedings
  of the 47th Annual Conference on Communication, Control and Computing
  (Allerton)}, 2009.

\bibitem{neely10}
M.~Neely, ``{Delay-based network utility maximization},'' in \emph{The 29th
  IEEE International Conference on Computer Communications (INFOCOM)}, 2010.

\bibitem{eryilmaz05}
A.~Eryilmaz, R.~Srikant, and J.~Perkins, ``{Stable scheduling policies for
  fading wireless channels},'' \emph{IEEE/ACM Transactions on Networking},
  vol.~13, no.~2, pp. 411--424, 2005.

\bibitem{liuj09}
J.~Liu, A.~Stolyar, M.~Chiang, and H.~Poor, ``{Queue Back-Pressure Random
  Access in Multihop Wireless Networks: Optimality and Stability},'' \emph{IEEE
  Transactions on Information Theory}, vol.~55, no.~9, pp. 4087 --4098, Sept.
  2009.

\bibitem{bui11}
L.~Bui, R.~Srikant, and A.~Stolyar, ``{A Novel Architecture for Reduction of
  Delay and Queueing Structure Complexity in the Back-Pressure Algorithm},''
  \emph{IEEE/ACM Transactions on Networking}, vol.~19, no.~6, pp. 1597--1609,
  2011.

\bibitem{gupta11}
G.~Gupta and N.~Shroff, ``{Delay analysis and optimality of scheduling policies
  for multihop wireless networks},'' \emph{IEEE/ACM Transactions on
  Networking}, vol.~19, no.~1, pp. 129--141, 2011.

\bibitem{sharma06}
G.~Sharma, R.~R. Mazumdar, and N.~B. Shroff, ``{On the complexity of scheduling
  in wireless networks},'' in \emph{Proceedings of the annual international
  conference on Mobile computing and networking (MobiCom)}.\hskip 1em plus
  0.5em minus 0.4em\relax ACM New York, NY, USA, 2006, pp. 227--238.

\bibitem{hajek88}
B.~Hajek and G.~Sasaki, ``{Link scheduling in polynomial time},'' \emph{IEEE
  Transactions on Information Theory}, vol.~34, no.~5, pp. 910--917, 1988.

\bibitem{joo09}
C.~Joo, X.~Lin, and N.~B. Shroff, ``{Greedy Maximal Matching: Performance
  Limits for Arbitrary Network Graphs Under the Node-exclusive Interference
  Model},'' \emph{IEEE Transactions on Automatic Control}, vol.~54, no.~12, pp.
  2734--2744, 2009.

\bibitem{joo09c}
C.~Joo and N.~B. Shroff, ``{Performance of random access scheduling schemes in
  multi-hop wireless networks},'' \emph{IEEE/ACM Transactions on Networking},
  vol.~17, no.~5, pp. 1481--1493, 2009.

\bibitem{lin06}
X.~Lin and N.~B. Shroff, ``{The impact of imperfect scheduling on cross-Layer
  congestion control in wireless networks},'' \emph{IEEE/ACM Transactions on
  Networking}, vol.~14, no.~2, pp. 302--315, 2006.

\bibitem{chaporkar08}
P.~Chaporkar, K.~Kar, X.~Luo, and S.~Sarkar, ``{Throughput and Fairness
  Guarantees Through Maximal Scheduling in Wireless Networks},''
  \emph{Information Theory, IEEE Transactions on}, vol.~54, no.~2, pp.
  572--594, 2008.

\bibitem{wu07}
X.~Wu, R.~Srikant, and J.~Perkins, ``{Scheduling Efficiency of Distributed
  Greedy Scheduling Algorithms in Wireless Networks},'' \emph{IEEE Transactions
  on Mobile Computing}, pp. 595--605, 2007.

\bibitem{leconte09}
M.~Leconte, J.~Ni, and R.~Srikant, ``{Improved bounds on the throughput
  efficiency of greedy maximal scheduling in wireless networks},'' in
  \emph{Proceedings of the tenth ACM international symposium on Mobile ad hoc
  networking and computing (MobiHoc)}.\hskip 1em plus 0.5em minus 0.4em\relax
  New York, NY, USA: ACM, 2009, pp. 165--174.

\bibitem{bramson08}
M.~Bramson, ``{Stability of queueing networks},'' \emph{Probability Surveys},
  vol.~5, no.~1, pp. 169--345, 2008.

\bibitem{rybko92}
A.~Rybko and A.~Stolyar, ``{Ergodicity of stochastic processes describing the
  operation of open queueing networks},'' \emph{Problems of Information
  Transmission}, vol.~28, pp. 199--220, 1992.

\bibitem{malyshev79}
V.~Malyshev and M.~Menshikov, ``{Ergodicity, continuity and analyticity of
  countable Markov chains},'' \emph{Transactions of the Moscow Mathematical
  Society}, vol.~39, pp. 3--48, 1979.

\bibitem{dai95}
J.~Dai, ``{On positive Harris recurrence of multiclass queueing networks: a
  unified approach via fluid limit models},'' \emph{The Annals of Applied
  Probability}, pp. 49--77, 1995.

\bibitem{stolyar95}
A.~Stolyar, ``{On the stability of multiclass queueing networks: a relaxed
  sufficient condition via limiting fluid processes},'' \emph{Markov Processes
  and Related Fields}, vol.~1, no.~4, pp. 491--512, 1995.

\bibitem{dimakis06}
A.~Dimakis and J.~Walrand, ``{Sufficient conditions for stability of
  longest-queue-first scheduling: second-order properties using fluid
  limits},'' \emph{Advances in Applied Probability}, vol.~38, no.~2, p. 505,
  2006.

\bibitem{joo09b}
C.~Joo, X.~Lin, and N.~B. Shroff, ``{Understanding the capacity region of the
  greedy maximal scheduling algorithm in multihop wireless networks},''
  \emph{IEEE/ACM Transactions on Networking}, vol.~17, no.~4, pp. 1132--1145,
  2009.

\bibitem{hoepman04}
J.~Hoepman, ``Simple distributed weighted matchings,'' \emph{Arxiv preprint
  cs/0410047}, 2004.

\bibitem{joo11}
C.~Joo and N.~B. Shroff, ``{Local Greedy Approximation for Scheduling in
  Multi-hop Wireless Networks},'' \emph{IEEE Transactions on Mobile Computing},
  accepted for publication.

\bibitem{brzezinski08}
A.~Brzezinski, G.~Zussman, and E.~Modiano, ``{Local pooling conditions for
  joint routing and scheduling},'' in \emph{Information Theory and Applications
  Workshop, 2008}, 2008, pp. 499--506.

\bibitem{zussman08}
G.~Zussman, A.~Brzezinski, and E.~Modiano, ``{Multihop Local Pooling for
  Distributed Throughput Maximization in Wireless Networks},'' in \emph{The
  IEEE International Conference on Computer Communications (INFOCOM)}, April
  2008, pp. 1139--1147.

\end{thebibliography}

% \newpage

\end{document}